\documentclass[11pt,a4paper]{article}

\usepackage{graphicx}
\graphicspath{{figures/}}
\usepackage[english]{babel}
\usepackage{latexsym}
\usepackage{url}
\usepackage{amsmath}
\usepackage{amsthm}
\usepackage{dsfont}
\usepackage{ifthen}
\usepackage{fancybox}
\usepackage{microtype,stmaryrd}

\usepackage[charter]{mathdesign}

\usepackage{color}
\definecolor{blu3}{rgb}{.1,.0,.4}
\usepackage{hyperref}
\hypersetup{colorlinks=true, linkcolor=blu3, urlcolor=blu3, citecolor=blu3, pdfpagemode=UseNone, pdfstartview=} %colored links, no rectangles, no bookmarks, zoom

\usepackage[margin=1.1in]{geometry}

\newtheorem{theorem}{Theorem}
\newtheorem{corollary}[theorem]{Corollary}
\newtheorem{lemma}[theorem]{Lemma}

\newcommand{\RR}{\ensuremath{\mathbb R}}  % real numbers
\DeclareMathOperator{\polylog}{polylog}

\DeclareMathOperator{\cell}{GVR}
\DeclareMathOperator{\GVD}{GVD}

\DeclareMathOperator{\bis}{bis}
\DeclareMathOperator{\interior}{int}
\DeclareMathOperator{\exterior}{ext}
\DeclareMathOperator{\AVD}{AVD}
\DeclareMathOperator{\AVR}{AVR}
\DeclareMathOperator{\AJ}{AJ}
\DeclareMathOperator{\AD}{AD}
\DeclareMathOperator{\GD}{GD}
\DeclareMathOperator{\diam}{diam}
\DeclareMathOperator{\counting}{count_\le}
\DeclareMathOperator{\adding}{sum}

\newcommand\numberthis{\addtocounter{equation}{1}\tag{\theequation}}
\newcommand\eps{\varepsilon}

\def\DEF#1{\textbf{\emph{#1}}}

\def\dart#1#2{#1\mathord\shortrightarrow#2}

\begin{document}

\title{Subquadratic Algorithms for the Diameter and \\
	the Sum of Pairwise Distances in Planar Graphs\thanks{A preliminary version of this work was presented at SODA 2017~\cite{Cab17}.}}

\author{Sergio Cabello\thanks{Department of Mathematics, IMFM, and
                Department of Mathematics, FMF, University of Ljubljana, Slovenia.
                Supported by the Slovenian Research Agency, program P1-0297. Email address: sergio.cabello@fmf.uni-lj.si}
}

\maketitle

\begin{abstract}
	We show how to compute for $n$-vertex planar graphs
	in $O(n^{11/6}\polylog(n))$ expected time the diameter and 
	the sum of the pairwise distances.
	The algorithms work for directed graphs with real weights and no negative cycles.
	In $O(n^{15/8}\polylog(n))$ expected time we can also compute 
	the number of pairs of vertices at distance smaller than a given threshold.
	These are the first algorithms for these problems using time $O(n^c)$ for some constant $c<2$, 
	even when restricted to undirected, unweighted planar graphs.

\end{abstract}

%%%%%%%%%%%%%%%%%%%%%%%%%%%%%%%%%%%%%%%%%%%%%%%%%%%%%%%%%%%%%%%%%%%%%%%%%%%%%%%%%%%%%%%%%%%%%%%%%%%%%%%%%%%%%%%%%%%%%%
\section{Introduction}
Let $G$ be a directed graph with $n$ vertices and 
arc-lengths $\lambda\colon E(G)\rightarrow \RR$.
The length of a walk in $G$ is the sum of the arc-lengths along
the walk. We assume that $G$ has no cycle of negative length.
The \DEF{distance} between two vertices $x$ and $y$ of $G$, denoted
by $d_G(x,y)$, is the minimum length over all paths in $G$ from $x$ to $y$. 
While it is common to use the term distance, this is not necessarily a metric.
This scenario is an extension of the more common case where the graph $G$
is undirected and the lengths are positive. In that case $d_G(\cdot,\cdot)$ is indeed a metric.

In this paper we are interested in computing basic information about the distances between vertices in the graph $G$.
The \DEF{diameter} of $G$ is 
\[
	\diam(G) ~:=~ \max \{ d_G(x,y) \mid x,y \in V(G)\},
\]
the \DEF{sum of the pairwise distances} of $G$ is
\[
	\adding(G) ~:=~ \sum_{(x,y)\in (V(G))^2} d_G(x,y),
\]
and, for any $\delta\in \RR$, the \DEF{distance counter} of $G$ is 
\[
	\counting(G,\delta) ~:=~ | \{(x,y)\in (V(G))^2 \mid d_G(x,y)\le \delta\}|.
\]

For undirected graphs, the value $\adding(G)$ is essentially equivalent to the average
distance in the graph and the so-called \emph{Wiener index}. The Wiener index is 
a basic topological index used in mathematical chemistry with thousands of publications. 

Computing the diameter, the sum of the pairwise distances, or the distance counter of 
a graph is a fundamental problem in graph algorithms.
The obvious way to compute them
is via solving the all-pairs shortest path problem (APSP) 
explicitly and then extract the relevant information.
A key question is whether one can avoid the explicit computation
of all the pairwise distances. 

Roditty and Vassilevska Williams~\cite{RodittyW13}
show that, for arbitrary graphs with $n$ vertices and $O(n)$ edges, 
one cannot compute the diameter in $O(n^{2-\delta_0})$ time, 
for some constant $\delta_0>0$, 
unless the strong exponential time hypothesis (SETH) fails.
In fact, their proof shows that for undirected, unweighted graphs
we cannot distinguish in $O(n^{2-\delta_0})$ time
between sparse graphs that have diameter $2$ or larger,
assuming the SETH. This implies the same conditional lower bound
for computing the sum of the pairwise distances or the distance counter in sparse graphs. 
Indeed, an unweighted graph $G$ of $n$ vertices has diameter $2$
if and only if 
\[
	\adding(G) ~=~ \sum_{x\in V(G)} \Bigl( \deg_G(x)+2(n-1-\deg_G(x)) \Bigr)
	~=~ 2n(n-1) - 2 |E(G)|.
\]	
Similarly such a graph $G$ has $\counting(G,2)=n^2$ if and only if  $G$ has diameter $2$.
Thus, if we could compute the sum of pairwise distances or the distance  counter for sparse graphs
in $O(n^{2-\delta_0})$ time, we could also distinguish in the same time
whether the graph has diameter $2$ or larger, and the SETH fails.

Given such conditional lower bounds, it is natural to 
shift the interest towards identifying families of sparse graphs 
where one can compute the diameter or the sum of pairwise distances in 
truly subquadratic time. Here we provide subquadratic algorithms for directed, planar graphs
with no negative cycles. 
More precisely, we show that the diameter and the sum of the pairwise distances can be computed in
$O(n^{11/6}\polylog(n))$ expected time,
while the distance counter can be computed in $O(n^{15/8}\polylog(n))$ expected time.
There are efficient algorithms for computing 
all the distances in a planar graph~\cite{f-pgdap-91} 
or a specified subset of the distances~\cite{Cabello12,MozesS12}.
However, none of these tools seem fruitful for computing our statistics in subquadratic time.

Note that our algorithms are the first  algorithms using time $O(n^c)$ for some constant $c<2$, 
even when restricted to undirected, unweighted planar graphs.
	
\paragraph{Related work.}
For graphs of bounded treewidth one can compute the
diameter and the sum of pairwise distances in
near-linear time~\cite{AWW16,CabelloK09}.
The distance counter for graphs of bounded treewidth can be handled using the same techniques.
Recently, Husfeldt~\cite{Husfeldt16} has looked at the problem of computing the diameter for undirected, 
unweighted graphs parameterized by the treewidth and the diameter.

For planar graphs,
Wulff-Nilsen~\cite{WN08} gives an algorithm to compute
the diameter and the sum of pairwise distances in unweighted, undirected planar graphs in 
$O(n^2 \log\log n/ \log n)$ time, which is slightly subquadratic.
Wulff-Nilsen~\cite{WN10} extends the result to weighted directed planar graphs with a time bound of
$O(n^2 (\log\log n)^4/ \log n)$.
Note that the running time of these algorithms is not 
of the type $O(n^c)$ for any constant $c < 2$.

Researchers have also looked into near-optimal approximations.
In particular, 
Weimann and Yuster~\cite{WeimannY16} provide a $(1+\eps)$-approximation 
to the diameter of undirected planar graphs in   
$O((n/\eps^4) \polylog(n) + 2^{O(1/\eps)} n)$ time.
As it was mentioned by Goldreich and Ron~\cite{GoldreichR08},
a near-linear time randomized $(1+\eps)$-approximation for the sum
of pairwise distances in undirected planar graphs can be obtained
using random sampling and an oracle for $(1+\eps)$-approximate 
distances~\cite{KKS11,Thorup04}.
See the work by Indyk~\cite{Indyk99} for the average distance in
arbitrary discrete metric spaces.

\paragraph{Our approach.}
Let us describe the high-level idea of our approach.
The main new ingredient is the use of additively-weighted Voronoi diagrams
in pieces of the graph: we make a 
quite expensive preprocessing step in each piece that 
permits the efficient computation of such Voronoi diagrams in each piece
for several different weights.

To be more precise, let $G$ be a planar graph with $n$ vertices. 
We first compute an $r$-division: this is a decomposition
of $G$ into $O(n/r)$ pieces, 
each of them with $O(r)$ vertices and $O(\sqrt{r})$ boundary vertices. 
This means that all the interaction between a piece $P$ and the complement
goes through the $O(\sqrt{r})$ boundary vertices of $P$.

Consider a piece $P$ and a vertex $x$ outside $P$. We would like to
break $P$ into regions according to the boundary vertex of $P$ that is used
in the shortest path from $x$. This can be modeled as an additively-weighted
Voronoi diagram in the piece: each boundary vertex is a weighted site whose
weight equals the distance from $x$. Thus, we have to compute several
such Voronoi diagrams for each piece. 

Assuming that a piece is embedded, one can treat such a Voronoi diagram
as an abstract Voronoi diagram and encode it using the dual graph. 
In particular, a bisector corresponds to a cycle in the dual graph. 
We can precompute all possible Voronoi diagrams for $O(1)$
sites, and that information suffices to compute
the Voronoi diagram using a randomized incremental construction.
Once we have the Voronoi diagram, encoded as a subgraph of the dual graph,
we have to extract the information from each Voronoi region.
Although this is the general idea, several technical details appear.
For example, the technology of abstract Voronoi diagrams can be used
only when the sites are cofacial.

We remark that our algorithms actually compute information for the distances from each vertex $x$ of $G$ separately.
Thus, for each vertex $x$ we compute the furthest vertex from $x$, the sum of the distances from $x$ to
all vertices, and the number of vertices at distance at most $\delta$ from $x$, for a given $\delta\in \RR$. 
Our main result is the following, whose statement makes this clear. 

\begin{theorem}
\label{thm:main}
	Let $G$ be a planar graph with $n$ vertices, real abstract length on its arcs, and no negative cycle.
	In $O(n^{11/6}\polylog(n))$ expected time we can compute $\adding(x,V(G),G)$ and
	$\diam(x,V(G),G)$ for all  vertices $x$ of $G$.
	For a given $\delta\in \RR$, in $O(n^{15/8}\polylog(n))$ expected time 
	we can compute $\counting(x,V(G),G,\delta)$ for all vertices $x$ of $G$.
\end{theorem}

The proof of Theorem~\ref{thm:main} is in Section~\ref{sec:together}.

\paragraph{Assumptions.} 
We will assume that the distance between each pair of vertices is distinct and
there is a unique shortest path between each pair of vertices.
This can be enforced with high probability using infinitesimal perturbations
or deterministically using lexicographic comparison;
see for example the discussion by Cabello, Chambers and Erickson~\cite{CabelloCE13}. 
Since our result is a randomized algorithm with running times
that are barely subquadratic, the actual method that is used is not very relevant.

\paragraph{Randomization.}
Our algorithm is randomized and it is good to explain the source of this.
Firstly, we use random perturbations of lengths of the edges to ensure unique shortest paths.
The author thinks that, with some work, this assumption could be removed.

Another source of randomization comes from our black-box use of the paper 
by Klein, Mehlhorn and Meiser~\cite{KleinMM93}.
They provide a randomized incremental construction of Voronoi diagrams
under very general assumptions. Randomized incremental constructions are a
standard tool in computational geometry.
At the very high level, we compute a random permutation 
$s_1,\dots,s_n$ of the sites that define the diagram, and then iteratively
compute the Voronoi diagram for the subsets $S_i=\{ s_1,\dots,s_i \}$.
To compute $S_i$ from $S_{i-1}$, one has to estimate the amount of changes
that take place, and this is a random variable. In the case of Voronoi diagrams,
this is related to the expected size of a face of the Voronoi diagram. 
Additional work is needed to keep pointers that allow to make the updates fast.
In particular, for the new site $s_i$, we have to find the current face of the Voronoi
diagram for $S_{i-1}$ that contains it.

\paragraph{Follow up work.}
Since the conference version of our paper there has been important progress using
some of the techniques introduced here.
Voronoi diagrams in planar graphs have been used
to construct distance oracles for planar graphs
that have subquadratic space and answer queries in logarithmic time~\cite{cdw17,gmww18}. 
Most importantly, Gawrychowski et al.~\cite{ghmsw18} provide a better understanding of the structure 
of Voronoi diagrams in planar graphs that leads to a deterministic construction 
with a faster preprocessing time. With this, they obtain faster and deterministic
algorithms for all the problems we consider here. While some of the ideas they use 
come from our work, they also provide several new, key insights.

\paragraph{Roadmap.}
We assume that the reader is familiar with planar graphs.
In the next section we explain the notation and some basic background.
In Section~\ref{sec:dualcycle} we explain how to extract information
about the vertices contained in a dual cycle. 
In Section~\ref{sec:AVD} we explain the concept of abstract Voronoi diagrams.
In Section~\ref{sec:GVD} we deal with different definitions of Voronoi diagrams 
in plane graphs and show that they are equivalent.
In Section~\ref{sec:algorithms} we discuss the algorithmic aspects of computing Voronoi diagrams. In
particular, the algorithm performs an expensive preprocessing to be able to produce Voronoi diagrams faster.
In Section~\ref{sec:datastructure} we give the data structure that will be used for each piece of an $r$-division.
In Section~\ref{sec:together} we give the final algorithms for planar graphs. 
We conclude with a discussion.

For some readers, it may be more pleasant to read Section~\ref{sec:together} 
before Sections~\ref{sec:dualcycle}-\ref{sec:datastructure}. 
This may help understanding the high level approach and how everything fits together
before delving into the details.

%%%%%%%%%%%%%%%%%%%%%%%%%%%%%%%%%%%%%%%%%%%%%%%%%%%%%%%%%%%%%%%%%%%%%%%%%%%%%%%%%%%%%%%%%%%%
\section{Notation and preliminaries}
\label{sec:preliminaries}
For running times, we use the notation $\tilde O(\cdot)$ when we omit polylogarithmic factors in any
of the parameters that appears in the statement. 
For example, if $n$ appears in the discussion, 
$\tilde O(mr)$ means $O(mr \log^c (mnr))$ for some constant $c$.

For each natural number $n$, we use the notation $[n]:=\{1,\dots,n\}$. 
For each set $A\subset \RR^2$, we use $\overline A$ for its closure and $A^\circ$ for its interior.

\paragraph{Graphs.}
Graphs considered in this paper are directed. We use $V(G)$ and $E(G)$ for the vertex and
the arc set of a graph $G$, respectively. We use the notation $\dart xy$ or $e$ to denote arcs. 
The tail of an arc $\dart xy$ is $x$, and $y$ is the head. We use $e^R$ for the reversal of the arc $e$. 
In some cases we may have parallel arcs. It should be clear from the context which arc we are referring to. 
When the orientation of the arc $\dart xy$ is not relevant, we may use $xy$ and refer to it as an (undirected) edge.

A closed walk in $G$ is a sequence $e_0,\dots , e_{k-1}$ of arcs with the property that the tail of $e_i$ is the head
of $e_{i-1}$ for all $i\in [k]$ (indices modulo $k$). Sometimes a closed walk is given as a sequence of vertices.
This uniquely defines the closed walk if there are no parallel edges. A cycle is a closed walk that
does not repeat any vertex. In particular, a cycle cannot repeat any arcs. We make it clear that the
walk $\dart xy,\dart yx$ is a cycle.

\paragraph{Planarity.}
A \DEF{plane graph} is a planar graph together with a fixed embedding. 
The arcs $e$ and $e^R$ are assumed to be embedded as a single curve with opposite orientations.
In the arguments we will use the geometry of the embedding and the plane quite often. 
For example, we will talk about the faces enclosed by a cycle of the graph.
However, all the computations can be done assuming a combinatorial
embedding, described as the circular order of the edges incident to each vertex. 

Let $G^*$ be the dual graph of a plane graph $G$. 
We may consider $G^*$ with oriented arcs or with edges, depending on the context.
We keep in $G^*$ any parallel edges that may occur.
When $G$ is 2-connected, the graph $G^*$ has no loops.
For each vertex $v$ and edge $e$ of $G$, 
we use $v^*$ and $e^*$ to denote their dual counterparts, respectively.
For any set of edges $A\subseteq E(G)$, 
we use the notation $A^*= \{ e^*\mid e\in A\}$.

We assume natural embeddings of $G$ and $G^*$ where each 
dual edge $e^*$ of $G^*$ crosses $G$ exactly once and does so at $e$.
There are no other types of intersections between $G$ and $G^*$.
See Figure~\ref{fig:dual} for an example.
If we would prefer to work with an actual embedding and coordinates, instead of a combinatorial
embedding, we could do so. To achieve this, for each edge $e$ of $G$, 
we subdivide $e$ and $e^*$ with a common vertex $v_e$. 
Then we obtain a planar graph $H$ that contains a subdivision of $G$ and a subdivision of $G^*$. 
We can now embed $H$ with straight-line segments in an $O(n)\times O(n)$ regular grid~\cite{TamassiaL04}. 
In this way we obtain an embedding of $G$ and an embedding of $G^*$ 
with the property that each edge and each dual edge is represented by a two-segment polygonal curve, 
and $e$ and $e^*$ cross as desired. With this
embedding we can carry out actual operations using coordinates.

\begin{figure}
\centering
	\includegraphics[page=2,scale=.9]{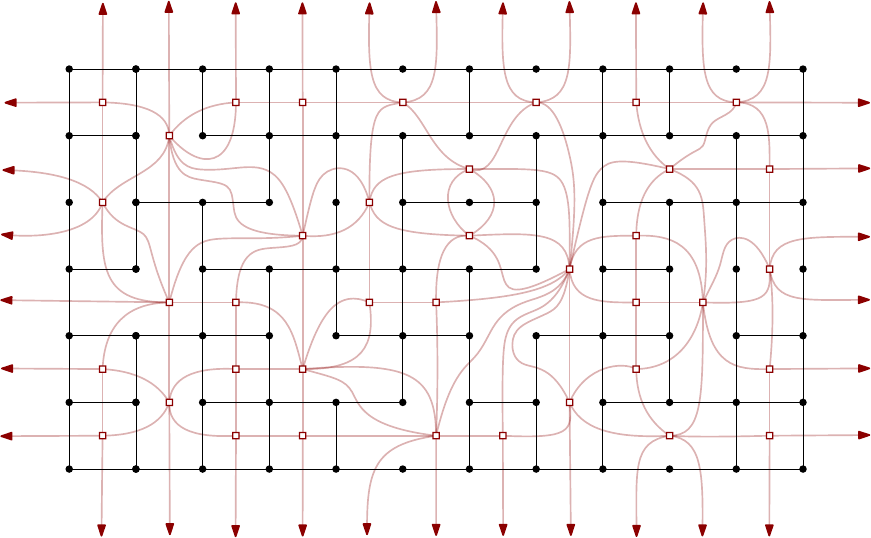}
	\caption{A plane graph $G$ (in black with dots for the vertices) and 
				its dual $G^*$ (in red with squares for the vertices). 
				The dual vertex $a_\infty$ corresponding to the outer face of $G$ is not drawn. 
				Dual edges with an endpoint at $a_\infty$ are represented using arrows.}
	\label{fig:dual}
\end{figure}

Vertices of $G$ are usually denoted by $x,y,u,v$.
Faces of $G$ are usually denoted by symbols like $f$ and $g$.
The dual vertices are usually denoted using early letters of the Latin alphabet, like $a$ and $b$.
We use $a_\infty$ for the dual vertex representing the outer face.
We will denote cycles and paths in the dual graph with Greek letters, 
such as $\gamma$ and $\pi$.
Sets of cycles and paths in the dual graph are with capital Greek letters, like $\Gamma$or $\Pi$.

Quite often we identify a graph object and its geometric representation
in the embedding. In particular, (closed) walks in the graph define (closed) 
curves in the plane. 
We say that a closed walk $\gamma$ in $G^*$ is \DEF{non-crossing}
if there is an infinitesimal perturbation $\gamma_\eps$ of the curve $\gamma$
that makes it simple. If $\gamma$ is simple, we can take $\gamma_\eps=\gamma$.
For each simple closed curve $\gamma$ in the plane,
let $\interior(\gamma)$ be the bounded domain of $\RR^2\setminus \gamma$,
and let $\exterior(\gamma)$ be the unbounded one. 
For each closed, non-crossing closed walk $\gamma$ in the dual graph $G^*$,
let $V_{\interior}(\gamma,G)=\interior(\gamma_\eps)\cap V(G)$ and
$V_{\exterior}(\gamma,G)=\exterior(\gamma_\eps)\cap V(G)$. 
Note that since $\gamma$ is a walk in $G^*$, the vertices of $V(G)$ are
far away from $\gamma$ and it does not matter which infinitesimal perturbation $\gamma_\eps$ 
of $\gamma$ we use. See Figure~\ref{fig:Vint} for an example.

\begin{figure}
\centering
	\includegraphics[page=3,scale=.9]{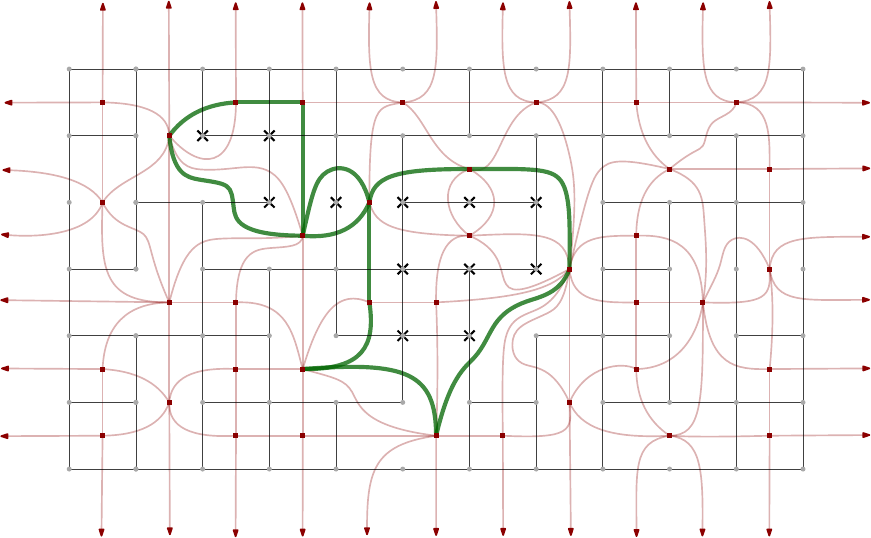}
	\caption{A non-crossing closed walk $\gamma$ in the graph of Figure~\ref{fig:dual}
			is drawn in thick green.
			The vertices of $V_{\interior}(\gamma,G)$ are marked with crosses.}
	\label{fig:Vint}
\end{figure}
\paragraph{Distances in graphs.}
In this paper we allow that the arcs have negative lengths $\lambda$. 
However, the graphs cannot have negative cycles, that is, cycles of negative length. 
In our approach we need that subpaths of shortest paths are also shortest paths.
Note that the existence of a cycle of negative length can be checked in near-linear time 
for planar graphs using algorithms for the shortest-path problem~\cite{fr-pgnwe-06,KleinMW10,MozesW10}.

For a graph $G$, a \DEF{shortest-path tree from} a vertex $r\in V(G)$ is a tree $T$ 
that is a subgraph of $G$ and satisfies $d_T (r,y) = d_G (r, y)$ for all $y\in V(G)$. 
A \DEF{shortest-path tree to} a vertex $r\in V(G)$ is a tree $T$ 
that is a subgraph of $G$ and satisfies $d_T (y,r) = d_G (y,r)$ for all $y\in V(G)$. 

For all graphs considered in this paper we assume that, whenever we have an arc $e$, we also have
its reversed arc $e^R$. We can ensure this by adding arcs with large enough length 
that no shortest path uses them.
Similarly, adding edges, we can assume that the graphs that we are considering are connected.

For a given graph $G$ with edge lengths $\lambda(\cdot)$, we use $G^R$ for the reversed graph, 
that is, the graph $G$ with edge lengths $\lambda^R(e) =\lambda (e^R)$. 
A shortest-path tree from $r$ in $G^R$ is the reversal of a shortest-path tree to $r$ in $G$.
Thus, as far as computation is concerned, there is no difference between computing shortest-path
trees from or to a vertex.

\paragraph{Potentials for directed graphs}
Let $G$ be a (directed) graph with arc lengths $\lambda\colon E(G)\rightarrow \RR$.
A \DEF{potential} for $G$ is a function $\phi\colon V(G)\rightarrow \RR$
such that:
\[
	\forall \dart uv \in E(G):~~~ \phi(v) \le \phi(u) + \lambda(\dart uv).
\]
For a potential function $\phi$ for $G$, the \DEF{reduced length} $\tilde\lambda$ is
defined by
\[
	\forall \dart uv \in E(G):~~~ \tilde\lambda(\dart uv) = \lambda( \dart uv) + \phi(u) - \phi(v).
\]
The following properties are easy and standard~\cite[Section 8.2]{Schrijver-book}. 
They have been used in several previous works in planar graphs.
\begin{itemize}
	\item Fix any vertex $s$ of $G$. If $G$ has no negative cycle, 
		then the function $\phi(v)=d_G(s,v)$ is a potential function.
	\item For each dart $\dart uv \in E(G)$ we have $\tilde\lambda(\dart uv)\ge 0$.
	\item A path in $G$ from $s$ to $t$ is a $\lambda$-shortest path 
		if and only if it is a $\tilde\lambda$-shortest path.
\end{itemize}
This means that, if $G$ has no negative cycle with respect to the arc lengths $\lambda$,
once we have computed a single-source shortest path tree in $G$ 
from an arbitrary source $s$, we can solve all subsequent single-source shortest path 
problems in $G$ using the reduced lengths, which are non-negative.

\paragraph{Vertex-based information.}
Consider a graph $G$. 
For each vertex $x\in V(G)$, each subset $U\subseteq V(G)$,
and each real value $\delta$, we define
\begin{align*}
	\diam(x,U,G) ~&:=~ \max \{ d_G(x,u)\mid u\in U \},\\
	\adding(x,U,G) ~&:=~ \sum_{u\in U} d_G(x,u) ,\\
	\counting(x,U,G,\delta) ~&:=~ | \{ u \in U\mid d_G(x,u)\le \delta \}|.
\end{align*}
Our main results will compute these values for all vertices $x\in V(G)$
when $G$ is planar and $U=V(G)$.
Clearly we have
\begin{align*}
	\diam(G) ~&=~ \max \{ \diam(x,V(G),G)\mid x\in V(G) \},\\
	\adding(G) ~&=~ \sum_{x\in V(G)} \adding (x,V(G),G) ,\\
	\counting(G,\delta) ~&=~ \sum_{x\in V(G)} \counting (x,V(G),G,\delta).
\end{align*}

%%%%%%%%%%%%%%%%%%%%%%%%%%%%%%%%%%%%%%%%%%%%%%%%%%%%%%%%%%%%%%%%%%%%%%%%%%%%%%%%%%%%%%%%%%%%
\section{Handling weights within a non-crossing walk}
\label{sec:dualcycle}

For the rest of this section, let $G$ be a plane graph with $n$ vertices.
In this section we are not concerned with distances. Instead, we are
concerned with vertex-weights.
Assume that each vertex $x$ of $G$ has a weight $\omega(x)\in \RR$.
For each subset $U$ of vertices and each value $\delta\in \RR$
we define 
\[
	\sigma(U) :=\sum_{x \in U} \omega(x), ~~~ 
	\mu(U) :=\max_{x \in U} \omega(x), ~~~
	\kappa_\le(U,\delta) := |\{  x\in U\mid \omega(x)\le \delta \}| .
\]

Let $\gamma$ be a non-crossing closed walk in the dual graph $G^*$. 
We are interested in a way to compute $\sigma(V_{\interior}(\gamma,G))$,
$\mu(V_{\interior}(\gamma,G))$, and $\kappa_\le(V_{\interior}(\gamma,G),\delta)$
\emph{locally}, after some preprocessing of $G$ and $G^*$.
Here, locally means that we would like to just look
at the edges of $\gamma$.
In the following, we assume that any non-crossing closed walk $\gamma$ in $G^*$ is 
traversed clockwise. 

In the next section we concentrate on the computation of $\sigma(\cdot)$ and then explain
how to use it for computing $\kappa_\le(\cdot,\delta)$.
In Section~\ref{sec:max} we discuss the computation of $\mu(\cdot)$
 
\subsection{Sum of weights and counting weights}
We start adapting the approach by Park and Phillips~\cite{ParkP93} and Patel~\cite{Patel13},
which considered the computation of $\sigma(\cdot)$ when $\omega(x)=1$ for all $x\in V(G)$.
We summarize the ideas in the next lemma to make it self-contained. 
While most of the paper is simpler for undirected graphs, in the next lemma we do
need the directed edges of the dual graph. We are not aware of a similar statement that
would work using the undirected dual graph.

\begin{lemma}
\label{le:sumweights}
	Let $G$ be a plane graph, directed or not, and let $x_0$ be a fixed vertex in $G$.
	In linear time we can compute a weight function 
	$\chi\colon E(G^*)\rightarrow \RR$ with
	the following property:
	For every non-crossing closed walk $\gamma$ in the dual graph $G^*$
	that is oriented clockwise and contains $x_0$ in its interior, 
	we have
	\[
		\sigma(V_{\exterior}(\gamma,G)) ~=~
			\sum_{\dart{a}{b} \in \gamma} \chi(\dart{a}{b}). 
	\]
\end{lemma}
\begin{proof}	
	Take any spanning tree $T$ of $G$ rooted at $x_0$, and orient the arcs away from $x_0$. 
	For example, a BFS tree of $G$ from $x_0$.
	For each node $y\in V(G)$, let $T_y$ be the subtree of $T$ rooted at $y$.
	See Figure~\ref{fig:sumweights}, left.
	For each vertex $y\neq x_0$ we proceed as follows. Let $x$ be the parent of $y$ and 
	let $\dart ab$ be the dual arc that crosses $\dart xy$ from left to right.
	Then we assign
	$\chi(\dart ab)= \sigma(V(T_y))$
	and $\chi(\dart ba)= - \chi(\dart ab)$.
	For any dual edge $ab$ of $E(G)^*\setminus E(T)^*$
	we set $\chi(\dart ab)=\chi(\dart ba)=0$.
	This finishes the description of the function $\chi$.
	It is easy to see that we can compute $\chi$ in linear time.
	
	From the definition of $\chi$ we have 
	\begin{align*}
		\sum_{\dart ab \in \gamma} \chi(\dart ab) ~&=~
		\sum_{\dart ab \in E(\gamma)\cap E(T)^*} \chi(\dart ab) \\
		&=~
		\sum_{\begin{minipage}{1.8cm}\centering \scriptsize
			$\dart xy \in T$,\\
			$\gamma$ crosses $\dart xy$\\ left-to-right
			\end{minipage}} \sigma(V(T_y)) -
		\sum_{\begin{minipage}{1.8cm}\centering \scriptsize
			$\dart xy \in T$,\\
			$\gamma$ crosses $\dart xy$\\ right-to-left
			\end{minipage}} \sigma(V(T_y))		. 
	\numberthis \label{eq:weights}
	\end{align*}

	\begin{figure}
		\centering
		\includegraphics[width=.7\textwidth]{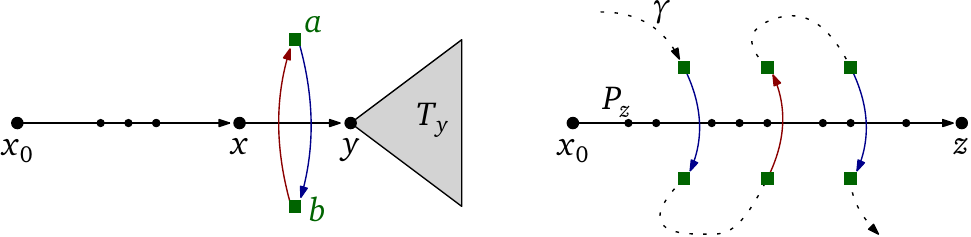}
		\caption{Proof of Lemma~\ref{le:sumweights}. 
			Left: orientation of the edges of $T$. 
			The dart $\dart ab$ crosses $\dart xy$ from left to right
			and thus $\chi(\dart ab)=\omega(V(T_y)$.
			Right: the crossings of $\gamma$ and $P_z$ alternate
			between left-to-right and right-to-left, as we walk along $P_z$.}
		\label{fig:sumweights}
	\end{figure}
	
	Let $\gamma_\eps$ be an infinitesimal perturbation of $\gamma$ that is simple.
	We then have $\interior(\gamma_\eps)\cap V(G) = V_{\interior}(\gamma,G)$
	and $\exterior(\gamma_\eps)\cap V(G) = V_{\exterior}(\gamma,G)$.
	
	Consider any vertex $z$ of $V(G)$ and let $P_z$ 
	be the path in $T$ from $x_0$ to $z$. Since $x_0$ is 
	in $\interior(\gamma_\eps)$ and $\gamma_\eps$ is a simple curve, 
	the crossings between $P_z$ and $\gamma_\eps$,
	as we walk along along $P_z$, alternate between left-to-right and
	right-to-left crossings.
	See Figure~\ref{fig:sumweights}, right.
	Since $\gamma_\eps$ defines a simple curve,
	the number of crossings is even if $z$ is in $\interior(\gamma_\eps)$
	and odd otherwise. It follows that $\omega(z)$ contributes to the sum
	on the right side of equation \eqref{eq:weights} either once, 
	if $z$ is in $\exterior(\gamma_\eps)$,
	or zero times, if $z$ is in $\interior(\gamma_\eps)$. The result follows.
\end{proof}

Lemma~\ref{le:sumweights} can also be used to compute $\sigma(V_{\interior}(\gamma,G))$
because $\sigma(V_{\interior}(\gamma,G))+ 
\sigma (V_{\exterior}(\gamma,G) =\sigma(V(G))$.

We would like a data structure to quickly handle non-crossing closed walks in
the dual graph that will be described compactly. 
More precisely, at preprocessing time 
we are given a family $\Pi=\{ \pi_1,\dots,\pi_\ell\}$ 
of walks in $G^*$, and the non-crossing closed walk will be given
as a concatenation of some subwalks from $\Pi$.
Using the function $\chi(\cdot)$ and partial sums over
the edges $e$ of each prefix of a walk in $\Pi$ we get the following result.

\begin{theorem}
\label{thm:sumweights}
	Let $G$ be a plane graph with $n$ vertices and vertex-weights $\omega(\cdot)$.
	Let $x_0$ be a vertex of $G$.
	Let $\Pi=\{ \pi_1,\dots,\pi_\ell\}$ be a family of walks in $G^*$ 
	with a total of $m$ edges, counted with multiplicity.
	After $O(n+m)$ preprocessing time, we can answer the following
	type of queries: 
	given a non-crossing closed walk $\gamma$ in $G^*$,
	described as a concatenation of $k$ subpaths of paths from $\Pi$,
	and with the property that $\gamma$ is oriented clockwise and 
	contains $x_0$ in its interior,
	return $\sigma (V_{\interior}(\gamma,G))$ 
	in $O(k)$ time.	
\end{theorem}
\begin{proof}
	We compute for $G$ the function $\chi$ of Lemma~\ref{le:sumweights}.
	For each walk $\pi_i$ of $\Pi$ we proceed as follows.
	Let $e(i,1),\dots,e(i,m_i)$ be the arcs of $\pi_i$, as they appear along the walk $\pi_i$,
	and define the partial sums 
	$S[i,j]=\sum_{t=1}^j \chi(e(i,t))$ for $j=1,\dots, m_i$.
	It is also convenient to define $S[i,0]=0$.
	The values $S[i,1], \dots, S[i,m_i]$ can be computed in $O(m_i)$ time
	using that $S[i,j]= S[i,j-1]+ \chi(e(i,j))$ for $j=1,\dots, m_i$.
	Repeating the procedure for each $\pi_i\in \Pi$, we have spent a total of $O(n+m)$ time.
	This finishes the preprocessing.

	Consider a non-crossing closed walk $\gamma$ in $G^*$
	given as the concatenation of $k$ walks $\pi^1,\dots,\pi^k$, 
	each of them a subpath of some path in $\Pi$. 
	Each $\pi^t$ in the description of $\gamma$ is of the form
	$e(i(t),j_1(t)),\dots, e(i(t),j_2(t))$ for some index $i(t)$ (so $\pi^t$ is a subpath of $\pi_{i(t)}$ )
	and some indices  $j_1(t),j_2(t)$ with $1\le j_1(t)\le j_2(t) \le m_{i(t)}$.
	Then we have 
	\[
		\sum_{e\in \pi^t} \chi(e) ~=~ S[i(t),j_2(t)] - S[i(t),j_1(t)-1].
	\]
	Because of the properties of $\chi$ in Lemma~\ref{le:sumweights}
	we have
	\begin{align*}
		\sigma(V_{\exterior}(\gamma,G)) ~=~
		\sum_{\dart{a}{b} \in \gamma} \chi(\dart{a}{b}) 
		~=~ \sum_{t=1}^k ~ \sum_{\dart{a}{b} \in \pi^t} \chi
			(\dart{a}{b}) \\
		~=~ \sum_{i=t}^k S[i(t),j_2(t)] - S[i(t),j_1(t)-1].
	\end{align*}
	It follows that we can compute $\sigma(V_{\exterior}(\gamma,G))$ in $O(k)$ time,
	and therefore we can also obtain $\sigma(V_{\interior}(\gamma,G))$ in the same time bound.
\end{proof}

We now look into the case of computing $\kappa_\le(\cdot,\delta)$.
Using a binary search on $W=\{ \omega(v)\mid v\in V(G)\}$ we achieve the following result.
Note that in the following result the dependency in $n$ increases.

\begin{corollary}
\label{cor:countweight}
	Consider the setting of Theorem~\ref{thm:sumweights}.
	After $O(n(n+m))$ preprocessing time, we can answer the following
	type of queries: 
	given a value $\delta\in \RR$ and
	a non-crossing closed walk $\gamma$ in $G^*$,
	described as a concatenation of $k$ subpaths of paths from $\Pi$,
	and with the property that $\gamma$ is oriented clockwise and 
	contains $x_0$ in its interior,
	return $\kappa_\le (V_{\interior}(\gamma,G),\delta)$ in $O(k+\log n)$ time.
\end{corollary}
\begin{proof}
	We sort the $n$ weights $W=\{ \omega(v)\mid v\in V(G)\}$ and store them in an array.
	Let $w_1,\dots,w_n$ be the resulting weights, so that $w_1\le \dots\le w_n$.
	For $i=1,\dots,n$, we define the weight function $\omega_i$ by
	\[
		\omega_i(v) ~=~\begin{cases}
					1 & \text{if $\omega(v)\le w_i$},\\
					0 & \text{otherwise}.
					\end{cases}
	\]
	Then, we apply Theorem~\ref{thm:sumweights} for each of the weight functions $\omega_1,\dots,\omega_n$.
	This finishes the preprocessing.
	
	To compute $\kappa_\le(V_{\interior}(\gamma,G),\delta)$ for a given $\delta\in \RR$,
	we make a a binary search in $W$ to find $w_i= \max \{ w\in W\mid w\le \delta\}$ and
	then use the data structure for the weight function $\omega_i$ to get
	\begin{align*}
		\sum_{v\in V_{\interior}(\gamma,G) } \omega_i(v)  ~&=~|\{ v\in V_{\interior}(\gamma,G) \mid \omega(v)\le w_i \}| \\
						&=~ \kappa_\le(V_{\interior}(\gamma,G),w_i) \\
						&=~ \kappa_\le(V_{\interior}(\gamma,G),\delta) .
	\end{align*}
	Thus, a query boils down to a (standard) binary search followed 
	by a single query to the data structure of Theorem~\ref{thm:sumweights}.
	Therefore the query time is $O(k+\log n)$.
\end{proof}

%%%%%%%%%%%%%%%%%%%%%%%%%%%%%%%%%%%%%%%%%%%%%%%%%%%%%%%%%%%%%%%%%%%%%%%%%%%%%%%%%%%%%%%%%%%%
\subsection{Maximum weight}
\label{sec:max}
The proof of Lemma~\ref{le:sumweights} heavily uses that the sum has an inverse operation.
We are not aware of any such result for computing the maximum weight,
$\mu(V_{\interior}(\gamma,G))$ or $\mu(V_{\exterior}(\gamma,G))$.
We could do something similar as we did in the proof of Corollary~\ref{cor:countweight},
namely, a binary search in $W=\{ \omega(v)\mid v\in V(G)\}$ to find
the largest weight inside $V_{\interior}(\gamma,G)$.
However, the extra preprocessing time in Corollary~\ref{cor:countweight},
as compared to the preprocessing time of Theorem~\ref{thm:sumweights},
leads to a worst running time in our target application.
We will develop now a different approach that works for a special type of closed walks
that we have in our application.

Let $x_0$ be a vertex of $G$ and let $T_0$ be a spanning tree of $G$ rooted at $x_0$. 
We say that a cycle $\gamma$ in the dual graph $G^*$ is \DEF{$T_0$-star-shaped}
if the root $x_0$ is in $\interior(\gamma)$ and,
for each vertex $y$ in $V_{\interior}(\gamma,G)$,
the whole path in $T_0$ from $x_0$ to $y$ is contained in $\interior(\gamma)$.
(Note that the concept is not meaningful for closed walks that repeat some vertex; 
hence our restriction to cycles for the time being.)
We define the following family of dual cycles:
\begin{equation}
	\Xi(G,T_0) ~=~ \{ \gamma\mid 
			\text{$\gamma$ is a $T_0$-star-shaped cycle in $G^*$}\}.
\label{eq:Xi}
\end{equation}

\begin{lemma}
\label{le:maxweight}
	There is a weight function 
	$\chi_{\mu}\colon E(G^*)\times E(G^*)\rightarrow \RR$ with
	the following properties:
	\begin{itemize}
	\item For every cycle $\gamma=e^*_0e^*_1,\dots ,e^*_{k-1}$ of $\Xi(G,T_0)$
			that is oriented clockwise 
			\[
				\mu(V_{\interior}(\gamma,G)) ~=~ \max \left
					\{ \chi_{\mu}(e^*_i,e^*_{i+1})\mid i=0,\dots ,k-1 \right \} ~~~~~~~
				\text{(indices modulo $k$)}
			\]
	\item After a linear-time preprocessing,  
		we can compute in constant time the value $\chi_{\mu}(ab,bc)$ 
		for any two dual edges $ab$ and $bc$ of $G^*$.
	\end{itemize}
\end{lemma}
\begin{proof}
	In this proof, for each vertex $v$, we use $T_0[\dart{x_0}{v}]$ to denote the 
	path in $T_0$ from $x_0$ to $v$.
	
	For a dual arc $e^*$, let $p(e^*)$ be the intersection point of $e$ and $e^*$,
	let $v(e^*)$ be the vertex of $e$ to the right of $e^*$,
	and let  $\pi(e^*)$ be the curve obtained by the concatenation of $T_0[\dart{x_0}{v(e^*)}]$ and
	the portion of $e$ from $v(e^*)$ to $p(e^*)$.
	See Figure~\ref{fig:maxweight1}, left.

	\begin{figure}
		\centering
		\includegraphics[page=1]{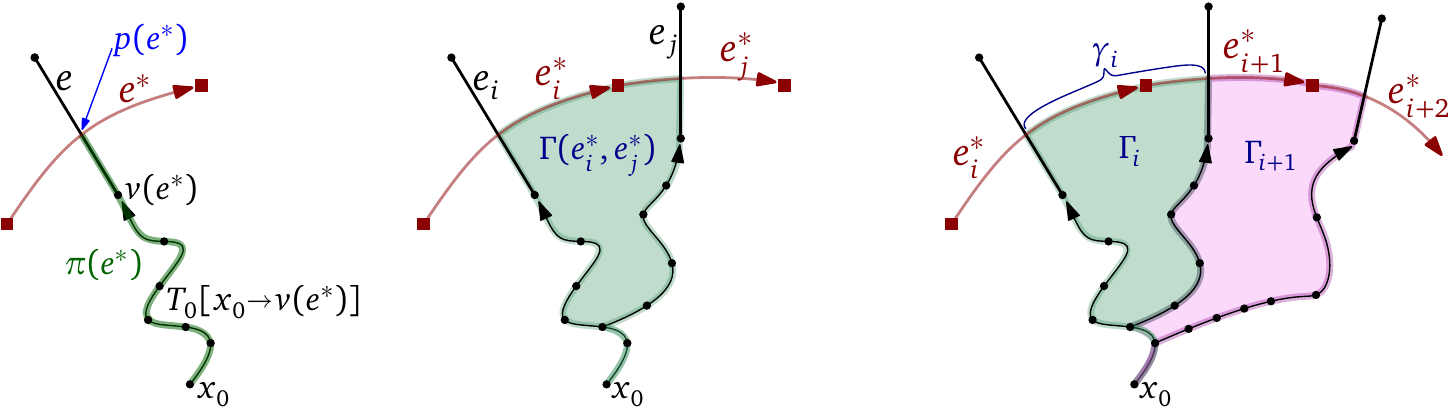}
		\caption{Defining $\chi_{\mu}(\cdot)$ in the proof of Lemma~\ref{le:maxweight}. 
			Left: Notation for a dual edge $e^*$. The path $\pi(e^*)$ is thicker.
			Middle: The region $\Gamma(e_i^*,e_j^*)$. We mark the portion of $T_0$ that
				is included in $\Gamma(e_i^*,e_j^*)$.
			Right: the regions $\Gamma_i$.}
		\label{fig:maxweight1}
	\end{figure}
	
	We can now provide a definition of the function $\chi_{\mu}$.
	Consider any two dual edges $e^*_i$ and $e^*_j$ in $G^*$. 
	If they have no common vertex or if they are equal, then we set $\chi_{\mu}(e^*_i,e^*_j)=0$.
	This is not very relevant because such terms never show up in the desired properties.
	It remains to consider the case when they have a common vertex.
	For this case we then define $\Gamma(e^*_i,e^*_j)$ as the region of the plane bounded 
	by $\pi(e^*_i)$, the portion of $e^*_i e^*_j$ from $p(e^*_i)$ to $p(e^*_j)$,
	and the reverse of $\pi(e^*_j)$. We regard the region $\Gamma(e^*_i,e^*_j)$ as a closed set,
	with its boundary and the curves that define it.
	See Figure~\ref{fig:maxweight1}, center.
	Finally, we set $\chi_{\mu}(e^*_i,e^*_j)= \mu \left( V(G)\cap \Gamma(e^*_i,e^*_j) \right)$.
	We will discuss the efficient computation and representation
	of $\chi_{\mu}(e^*_i,e^*_j)$ later.
	
	We claim that $\chi_{\mu}$ satisfies the property in the first item.
	Consider any dual cycle $\gamma=e^*_0e^*_1,\dots ,e^*_{k-1}$ and let $A$ be the closure of $\interior(\gamma)$.
	For each $i$, let $\gamma_i$ be the curve described by $\gamma$ from $p(e^*_i)$ to $p(e^*_{i+1})$
	and let us use the shorthand $\Gamma_i=\Gamma(e^*_i,e^*_{i+1})$, where indices are modulo $k$.
	See Figure~\ref{fig:maxweight1}, right.	
	Note that $\gamma_i$ is one of the curves used to define the region $\Gamma_i$.
	If $\gamma$ is in $\Xi(G,T_0)$, then the three curves that bound $\Gamma_i$ are contained in $A$ 
	and therefore $\Gamma_i$ is contained in $A$ ($i=0,\dots, k-1$).
	Moreover, since $\gamma_0,\dots,\gamma_{k-1}$ is a decomposition of $\gamma$ 
	and the regions $\Gamma_i$ and $\Gamma_{i+1}$ ($i=0,\dots, k-1$)
	share the path $\pi(e^*_{i+1})$ on its boundary, the union $\cup_i \Gamma_i$ is precisely $A$.
	Since the boundary of $A$ does not contain any vertex of $G$, we get
	\begin{align*}
			\mu(V_{\interior}(\gamma,G)) ~&=~ 
					\max \left\{ \omega(v)\mid v\in V(G)\cap \interior(\gamma) \right\}  \\
					&=~ \max  \left\{ \mu(V(G)\cap \Gamma_i)\mid i=0,\dots, k-1 \right\}\\
					&=~ \max  \left\{ \mu(V(G)\cap \Gamma(e^*_i,e^*_{i+1})\mid i=0,\dots, k-1 \right\} &
				\text{(indices modulo $k$)}\\
 					&=~ \max  \left\{\chi_{\mu}(e^*_i,e^*_{i+1}) \mid i=0,\dots, k-1 \right\}. &
				\text{(indices modulo $k$)}
	\end{align*}
	We have shown that the property in the first item holds. 
	
	It remains to discuss the computational part. First we discuss an alternative definition of $\chi_{\mu}$
	that is more convenient for the computation.
	For each edge $uv$ of $G$, let $R(uv)$ be the region of the plane defined by 
	the paths in $T_0$ from $x_0$ to both endpoints of $uv$ and the edge $uv$ itself.
	We include in $R(uv)$ the two paths used to define it and the edge $uv$.
	If the paths in $T_0$ from $x_0$ to $u$ and $v$ share a part, then the region $R(uv)$
	also contains that common part. 
	If $uv$ is in $T_0$, then the region $R(uv)$ is actually a path contained in $T_0$.
	See Figure~\ref{fig:maxweight2} for an example.
	Finally, for each edge $uv$ of $G$ we define $\varphi(uv)$ as
	\[
		\varphi(uv) ~:=~ \mu (V(G)\cap R(uv)) ~=~ \max \{  \omega(x) \mid x\in V(G)\cap R(uv) \}.
	\]
	For each vertex $v$ of $G$, we define $\varphi(v)$ as the maximum
	weight on the path $T_0[\dart{x_0}{v}]$. 
	This last case can be interpreted as a degenerate case of the previous one.
	Indeed, if $v'$ is the parent of $v$ in $T_0$, then $\varphi(v)=\varphi(vv')$.

	\begin{figure}
	\centering
		\includegraphics[page=2,width=\textwidth]{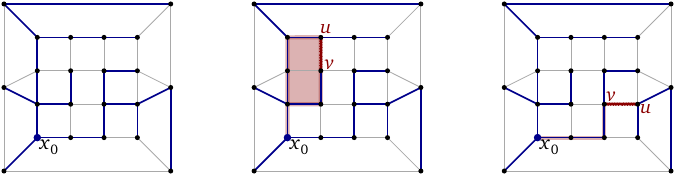}
		\caption{Example of the regions $R(uv)$.
				Left: a graph $G$ and a spanning tree $T_0$.
				Center and right: the regions $R(uv)$ for two different edges $uv$.}
		\label{fig:maxweight2}
	\end{figure}
	
	Let $f$ be a face of $G$ and let $e_i$ and $e_j$ be two edges on the boundary of $f$.
	We are going to give an alternative definition of $\chi_{\mu}(e^*_i,e^*_j)$.
	If $e_j$ is \emph{not} the follower of $e_i$ along  the counterclockwise traversal of $f$, 
	let $E(f,e_i,e_j)$ be the edges between $e_i$ and $e_j$ in a counterclockwise traversal of $f$. 
	We do not include $e_i$ and $e_j$ in $E(f,e_i,e_j)$, but the set $E(f,e_i,e_j)$ is nonempty by assumption.
	See Figure~\ref{fig:maxweight4} for an illustration.
	In this case we have 
	\begin{equation}
		\chi_{\mu}(e^*_i,e^*_j) ~=~ \max \{ \varphi (e)\mid e\in  E(f,e_i,e_j) \} .
		\label{eq:chi_mu}
	\end{equation}
	To see that this equality indeed holds, note that the difference between
	the region $\Gamma(e^*_i,e^*_j)$ and $\bigcup \{R (e)\mid e\in  E(f,e_i,e_j)\}$
	is just a portion of  the interior of the face $f$, which cannot contain vertices of $G$. 
	See Figure~\ref{fig:maxweight4}, center and right, for an illustration.
	If $e_i$ and $e_j$ are consecutive along  the counterclockwise traversal of $f$,
	then they have a common vertex $v$ and we have $\chi_{\mu}(e^*_i,e^*_j)=\varphi(v)$.
	The argument in this case is the same: the difference between  $\Gamma(e^*_i,e^*_j)$
	and the path $T_0[\dart{x_0}v]$ is a portion of the interior of $f$.

	\begin{figure}
	\centering
		\includegraphics[page=4]{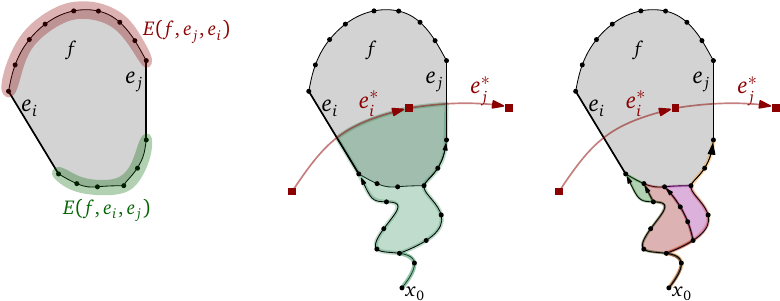}
		\caption{Left: example for the sets $E(f,e_i,e_j)$. 
			Center and right: example indicating the relation between $\Gamma(e^*_i,e^*_j)$
			and $\bigcup \{ R(e)\mid e\in E(f,e_i,e_j)\}$. (Compare to Figure~\ref{fig:maxweight1}.)
			Note that on the right we have five regions $R(e)$, but two of them degenerate
			to paths on $T_0$.}
		\label{fig:maxweight4}
	\end{figure}
	
	The second, alternative definition of $\chi_{\mu}$ is more suitable for efficient management.
	First, we compute the values $\varphi(\cdot)$. For this we use the undirected version of $G$.
	Let $C=E(G)\setminus E(T_0)$ be the primal edges not contained in $T_0$.
	The duals of those edges, $C^*$, form a spanning tree of the dual graph $G^*$.
	The pair $(T_0,C)$ is a so-called tree-cotree decomposition.
	We root $C^*$ at the dual vertex representing the outer face of $G$.
	Each edge $e\in C$ defines a region $A_e$ of the plane, namely the closed region bounded
	by the unique simple closed curve contained in $T_0+e$.
	Note that, for each $uv\in C$, 
	the region $R(uv)$ is precisely the union of $A_{uv}$ and
	the two paths  in $T_0$ from $x_0$ to the endpoints of $uv$.
	Each edge $e\in C$ defines a dual subtree, 
	denoted by $C^*_e$, which is the component of $C^*-e^*$
	without the root.
	The region $A_e$ corresponds to the faces of $G$ 
	that dualize to vertices of $C^*_e$.
	See Figure~\ref{fig:maxweight3} for an example.
	After computing and storing for each face of $G$ the maximum weight
	of its incident vertices, we can use 
	a bottom-up traversal of the dual tree $C^*$ and the values stored
	for each face to compute the values 
	$\mu\left(A_e\cap V(G)\right)$ in linear time 
	for all edges $e\in C$. 
	With a top-bottom traversal of the primal tree $T_0$ we can
	also compute and store for each node $v$ of $G$ the values 
	$\mu\left(T_0[\dart{x_0}{v}] \cap V(G)\right)$.
	From this we can compute $\varphi()$ as follows:
	\begin{align*}
		\forall v\in V(G)       &:~~~ \varphi(v)= \mu\left(T_0[\dart{x_0}{v}] \cap V(G) \right),\\
		\forall uv\in E(T_0) &:~~~ \varphi(uv)= \max \left\{  \mu\left(T_0[\dart{x_0}{v}] \cap V(G)\right), \omega(u),\omega(v) \right\},\\
		\forall uv\in C         &:~~~ \varphi(uv)= \max\left\{ \mu\left(A_{uv}\cap V(G)\right) ,  
				\mu\left(T_0[\dart{x_0}{u}] \cap V(G)\right), \mu\left(T_0[\dart{x_0}{v}] \cap V(G)\right) \right\}.
	\end{align*}
	Since each value on the right side is already computed, we spend linear time to compute the values $\varphi(\cdot)$.
	
	\begin{figure}
	\centering
		\includegraphics[page=3]{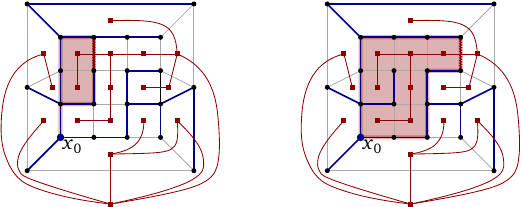}
		\caption{Dual tree defined by the cotree $C$ and its relevance to computing $\varphi(uv)$.
			Here we show the relevance for two different edges $uv\in (E(G)\setminus E(T_0))^*$
			in the example of Figure~\ref{fig:maxweight3}.}
		\label{fig:maxweight3}
	\end{figure}

	To represent $\chi_{\mu}$ compactly, we will use a data structure for range minimum queries:
	preprocess an array of numbers $A[1\dots m]$ such that, at query
	time, we can report $\min\{ A[k]\mid i\le k\le j\}$ for any given query pair of indices $i<j$.
	There are data structures that use linear-time preprocessing and $O(1)$ time per 
	query~\cite{BenderFPSS05,FischerH07}. This data
	structure does exploit the full power of random-access memory (RAM). It is trivial
	to extend this data structure for circular arrays: each query in  a circular array corresponds
	to two queries in a linear array.
	
	For each face $f$ of $G$, we build a circular array $A_f[\cdot]$ indexed by the edges, as they
	appear along the face $f$. At the entry $A_f[e]$ we store the value $\varphi(e)$.
	For each face we spend time proportional to the number of edges on the boundary of the face.
	Thus, for the whole graph $G$ this preprocessing takes linear time.
	For two edges $e_1$ and $e_2$ on the boundary of a face $f$ and with no common vertex, 
	the value $\chi_{\mu}(e_1^*,e_2^*)$, as described in \eqref{eq:chi_mu},
	is precisely a range maximum query in the circular array $A_f[\cdot]$,
	and thus can be answered in constant time. The case when $e_1$ and $e_2$ have a common
	vertex is easier because for each edge $e_1$ there are only two such possible edges $e_2$, 
	one per face with $e_1$ is on the boundary.
\end{proof}

For our application we will have to deal with pieces that have holes and thus a part of $T_0$
may be missing. 
Because of this, we also need to extend things to a type of non-crossing walks.

Like before, let $G$ be a plane graph and let $T_0$ be a rooted spanning subtree.
Let $P$ be a subgraph of $G$, with the embedding inherited from $G$.
Assume that the root $x_0$ of $T_0$ is in $P$.
A non-crossing closed walk $\gamma$ in $P^*$ is \DEF{$T_0$-star-shaped}
if the root $x_0$ is in $\interior(\gamma)$ and,
for each vertex $y$ in $V_{\interior}(\gamma,P)$,
all the vertices of $P$ in the path $T_0[\dart{x_0}{y}]$ are contained in $\interior(\gamma)$.
We can define the following family of dual non-crossing walks:
\begin{equation}
	\widetilde\Xi(G,P,T_0) ~=~ \{ \gamma\mid 
			\text{$\gamma$ is a $T_0$-star-shaped non-crossing walk in $P^*$}\}.
\end{equation}
Thus, each non-crossing walk in $\widetilde\Xi(G,P,T_0)$ comes from some cycle of $\Xi(G,T_0)$
when we transform $G$ into $P$ by deleting the edges of $E(G)\setminus E(P)$.

Let us provide some intuition for the following statement.
Consider a plane graph $G$ and a spanning tree $T_0$ of $G$.
Now we delete some of the edges of $G$ until we get a subgraph $P$, without changing the embedding.
We may have deleted some edges of $T_0$ also. However, the root of $T_0$ remains in $P$.
Some faces of $P$ may contain some of the edges of the spanning tree, $E(T_0)$, that were deleted. 
That is, when we draw an edge $e\in E(T_0)\setminus E(P)$ back in its original position,
it is contained in the closure of a $f$ face of $P$. In such a case we say that the interior of $f$ 
intersects $T_0$.
We use $b$ for the sum, over the faces $f$ of $P$ whose interior intersects $T_0$,
of the number of edges of $P$ on the boundary of $f$. Thus, for each face $f$ in the sum, we 
count how many edge of $P$ define the face. 

\begin{theorem}
\label{thm:maxweight}
	Let $G$ be a plane graph with $n$ vertices and vertex-weights $\omega(\cdot)$,
	and let $T_0$ be a rooted spanning tree in $G$.
	Let $P$ be a subgraph of $G$ such that the root of $T_0$ is a vertex of $P$.
	Let $b$ be the number of edges of $P$ on all the faces of $P$ whose interior intersects $E(T_0)$.
	Let $\Pi=\{ \pi_1,\dots,\pi_\ell\}$ be a family of walks in $P^*$ with a total of $m$ edges, counted with multiplicity.
	After $O(n+m+b^3)$ preprocessing time, we can answer the following
	type of queries: 
	given a closed walk $\gamma$ in $\widetilde\Xi(G,P,T_0)$, 
	described as a concatenation of $k$ subpaths of paths from $\Pi$
	and oriented clockwise 
	return $\mu(V_{\interior}(\gamma,P))$ in $O(k)$ time.			
\end{theorem}
Although the dependency on $b$ in the time bound can perhaps be reduced,
it is sufficient for our purposes because currently the bottleneck is somewhere else.
\begin{proof}
	We may assume that $V(G)=V(P)=V(T_0)$. To see this, first we note that
	we can remove edges of $G$ that are not in $E(T_0)\cup E(P)$ because
	they do not play any role. Then we can replace each maximal subtree of $T_0-V(P)$ 
	by edges that connect vertices of $P$ without changing
	the set $\widetilde\Xi(G,P,T_0)$. For this we just need the ancestor-descendant relation
	between vertices incident to the face.
	See Figure~\ref{fig:maxweight5} to see the transformation.
	Thus, from now on we restrict ourselves
	to the case where $V(G)=V(P)=V(T_0)$.

	\begin{figure}
	\centering
		\includegraphics[page=5]{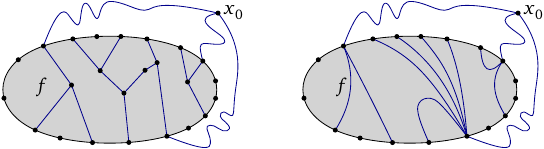}
		\caption{Transformation in the proof of Theorem~\ref{thm:maxweight}
			to assume that $V(G)=V(P)=V(T_0)$.
			The vertices of $V(T_0)$ inside a face $f$ of $P^*$ can be removed and
			we use direct edges representing the paths in $T_0$.}
		\label{fig:maxweight5}
	\end{figure}
	
	Let $F$ be the set of faces of $P$ that contain some edges of $E(T_0)\setminus E(P)$
	and consider the set 
	\[
		A~=~ \{ (e_1,e_2)\in E(P)^2 \mid \text{$e_1,e_2\in E(f)$ for some $f\in F$}\}.
	\]
	It is clear that $A$ has $O(b^2)$ pairs.
	For each $(e_1,e_2)\in A$, let $f$ be the face of $F$ that have $e_1$ and $e_2$ on the boundary
	and compute a dual path $\pi_{G}(e^*_1,e^*_2)$ in $G^*$ 
	from $e_1^*$ to $e_2^*$ whose other edges are contained in the face $f$.
	This means that all the edges of $\pi_{G}(e^*_1,e^*_2)$, except $e^*_1$ and $e^*_2$, are edges of $E(T_0)^*\setminus E(P)^*$.
	(If $e_1$ and $e_2$ are cofacial in $G$, then $\pi_{G}(e^*_1,e^*_2)=e^*_1 e^*_2$.)
	See Figure~\ref{fig:maxweight6} for an example.
	In particular, since $E(T_0)\setminus E(P)$ is a forest on $b$ vertices, it has at most $b$ edges, 
	and the path $\pi_{G}(e^*_1,e^*_2)$ has $O(b)$ edges.
	Thus, the paths $\pi_{G}(e^*_1,e^*_2)$, over all elements $(e_1,e_2)\in A$, have together $O(b^3)$ edges.

	\begin{figure}
	\centering
		\includegraphics[page=6,width=\textwidth]{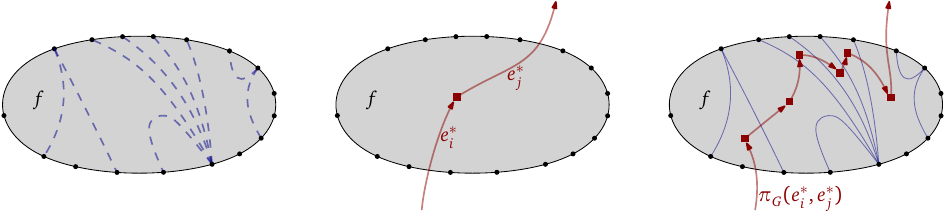}
		\caption{Construction of the paths $\pi_{G}(e^*_1,e^*_2)$. 
			Left: a face of $P$ with edges of $T_0-E(P)$ dashed.
			Center: a two-edge walk in $P^*$. 
			Right: the corresponding path $\pi_{G}(e^*_1,e^*_2)$ in $G^*$.}
		\label{fig:maxweight6}
	\end{figure}
 
	The paths $\{ \pi_{G}(e^*_1,e^*_2)\mid (e_1,e_2)\in A \}$ are used to naturally 
	transform walks in $P^*$ into walks  in $G^*$.
	Indeed, if we have a walk $\alpha$ in $P^*$ and we replace each occurrence of $e^*_1 e^*_2$, where $(e_1,e_2)\in A$
	by $\pi_{G}(e^*_1,e^*_2)$, then we obtain a walk in $G^*$.
 
	We compute for $G$ the function $\chi_{\mu}$ of Lemma~\ref{le:maxweight}.
 	For each element $(e_1,e_2)\in A$, we compute and store
 	\[
		\varphi(e^*_1,e^*_2) ~=~ \max\{ \chi_{\mu}(ab,bc) \mid \text{$ab$ and $bc$ consecutive dual edges along $\pi_{G}(e^*_1,e^*_2)$} \}.
	\]
 	Using the  properties of $\chi_{\mu}$ stated in Lemma~\ref{le:maxweight} and 
	using that the paths $\{ \pi_{G}(e^*_1,e^*_2)\mid (e_1,e_2)\in A\}$ have 
	$O(b^3)$ edges, we can do this step in $O(n+b^3)$ time.
 	
	For each walk $\pi_i$ of $\Pi$ we proceed as follows.
	Let $e^*_1,\dots,e^*_{m_i}$ be the edges of $\pi_i$, as they appear along $\pi_i$.
	We make an array $A_{i}[1..(m_i-1)]$ such that 
	\[
		A_{i}[j] ~=~ \begin{cases}
				\varphi(e^*_j,e^*_{j+1}) & \text{if $(e_j,e_{j+1})\in A$},\\
				\chi_{\mu}(e^*_j,e^*_{j+1}) & \text{if $(e_j,e_{j+1})\notin A$}.
				\end{cases}
	\]
	Finally, we store the array $A_i[\cdot]$ for range maximum queries~\cite{BenderFPSS05}; see the discussion at the end of 
	the proof of Lemma~\ref{le:maxweight}. We spend $O(m_i)$ preprocessing time for $\pi_i$
	and can find $\min A[j..j']$ in constant time for any given indices $1\le j<j'< m_i$.
	This step, together for all paths $\pi_i\in \Pi$, takes $O(\sum_i m_i)=O(m)$ time.
	This finishes the preprocessing. 

	Assume that we are given a non-crossing closed walk $\gamma$ in $\widetilde\Xi(G,P,T_0)$,
	given as the concatenation of $k$ paths $\pi^1,\dots,\pi^k$, 
	each of them a subpath of some path in $\Pi$. 
	Let $\gamma_G$ be the closed walk obtained from $\gamma$ as follows: for each $(e_1,e_2)\in A$
	and each appearance of $e^*_1 e^*_2$ in $\gamma$, we replace  $e^*_1 e^*_2$ by $\pi_G(e^*_1,e^*_2)$. 
	Note that $\gamma_G$ is a closed walk in $G^*$.
	In fact, $\gamma_G$ is a cycle in $G^*$ because geometrically each single replacement occurs within a single face of $F$
	and all the replacements within a face do not introduce crossings because $\gamma$ was non-crossing. 
	Moreover, because each replacement is a rerouting within a face $F$ and $V(G)=V(P)$,
	we have $V_{\interior}(\gamma,P)=V_{\interior}(\gamma_G,G)$.
	From Lemma~\ref{le:maxweight} we thus get that
	\begin{align*}
		 \mu(V_{\interior}(\gamma,P))~&=~ \mu(V_{\interior}(\gamma,G)) \\
		&=~ \max \{ \chi_{\mu}(ab,bc)\mid \text{$ab$ and $bc$ consecutive dual edges along the cycle $\gamma_G$} \}.
	\end{align*}
	Like in the proof of Theorem~\ref{thm:sumweights}, we can break the computation
	of $\chi_{\mu}(\cdot)$ for pairs of consecutive edges of $\gamma_G$ into $k$ parts that occur within some path $\pi_i\in \Pi$ (after replacements)
	and $k$ parts that use the last edge of $\pi^t$ and the first of $\pi^{t+1}$ ($t=0,\dots,k$, indices modulo $k$).
	The part within a path $\pi_i\in \Pi$ can be retrieved in constant time from the range maximum query for  $A_i[\cdot]$.
	The part combining consecutive subpaths can be computed in constant time, but we have two cases to consider.
	Let $ab$ be the last dual edge of $\pi^t$ and let $bc$ be the first dual edge of $\pi^{t+1}$.
	If $(ab,bc)\in A$, then we have to use $\varphi(ab,bc)$. Otherwise we can directly use 
	$\chi_{\mu}(ab,bc)$, which can be computed in constant time (second item of Lemma~\ref{le:maxweight}).
	Finally, we have to take the maximum from those $2k$ values.
\end{proof}

When $G=P$, then $\widetilde\Xi(G,P,T_0)=\Xi(G,T_0)$, we have $b=0$, and
Theorem~\ref{thm:maxweight} simplifies to the following.

\begin{corollary}
\label{coro:maxofweights}
	Let $G$ be a plane graph with $n$ vertices and vertex-weights $\omega(\cdot)$,
	and let $T_0$ be a rooted spanning tree in $G$.
	Let $\Pi=\{ \pi_1,\dots,\pi_\ell\}$ be a family of paths in $G^*$ 
	with a total of $m$ edges, counted with multiplicity.
	After $O(n+m)$ preprocessing time, we can answer the following
	type of queries: 
	given a cycle $\gamma$ in $\Xi(G,T_0)$, 
	described as a concatenation of $k$ subpaths of paths from $\Pi$
	and oriented clockwise 
	return $\mu(V_{\interior}(\gamma,G))$ in $O(k)$ time.	
\end{corollary}

%%%%%%%%%%%%%%%%%%%%%%%%%%%%%%%%%%%%%%%%%%%%%%%%%%%%%%%%%%%%%%%%%%%%%%%%%%%%%%%%%%%%%%%%%%%%
\section{Abstract Voronoi diagrams}
\label{sec:AVD}

Abstract Voronoi diagrams were introduced by Klein~\cite{Klein89} as
a way to handle together several of the different types of Voronoi diagrams
that were appearing. The concept is restricted to the plane $\RR^2$.
They are defined using the concept of \emph{bisectors} and 
\emph{dominant regions}.  
We will use the definition by Klein, Langetepe and Nilforoushan~\cite{KleinLN09},
as it seems the most recent and general.
For the construction, we use the randomized incremental construction
of Klein, Mehlhorn and Meiser~\cite{KleinMM93}, 
also discussed by Klein, Langetepe and Nilforoushan~\cite{KleinLN09} for their framework.
In our notation, we will introduce an $A$ in front to indicate
we are talking about objects in the \emph{abstract} Voronoi diagram.

Let $S$ be a finite set, which we refer to as \emph{abstract sites}.
For each ordered $(p,q)\in S^2$ of distinct sites, 
we have a simple planar curve $\AJ(p,q)$ and 
an open domain $\AD(p,q)$ whose boundary is $\AJ(p,q)$.
We refer to the pair $(\AJ(p,q),\AD(p,q))$ as an \DEF{abstract bisector}.
Define for each $p\in S$ the \DEF{abstract Voronoi region} 
$\AVR(p,S)=\bigcap_{q\in S\setminus\{ p\}} \AD(p,q)$.
Then the \DEF{abstract Voronoi diagram} of $S$, 
denoted by $\AVD(S)$, is defined as
$\AVD(S)=\RR^2\setminus \bigcup_{p\in S} \AVR(p,S)$.

The intuition is that the set $\AD(p,q)$ 
is the set of points that are closer to $p$ than
to $q$ and that $\AJ(p,q)$ plays the role of bisector.
Then, $\AVR(p,S)$ stands for the points that are dominated by $p$, when
compared against all $q\in S\setminus\{ p\}$. Note that $\AVR(p,S)$ is an open set because
it is the intersection of open sets. The abstract Voronoi diagram, $\AVD(S)$
would then be the set of points where no site dominates, 
meaning that at least two sites
are ``equidistant" from the point.
However, the theory does not rely on any such interpretations.
This makes it very powerful but less intuitive:  
some arguments become more cumbersome.

While these concepts can be considered in all generality,
the theory is developed assuming that certain properties, called axioms, are satisfied.
A system of abstract bisectors 
$\{ (\AJ(p,q),\AD(p,q))\mid  p,q\in S, p\not= q\}$ is \DEF{admissible}
if it satisfies the following properties:
\begin{enumerate}
	\item[(A1)] For all distinct $p,q\in S$, $J(p,q)=J(q,p)$.
	\item[(A2)] For all distinct $p,q\in S$, the plane $\RR^2$ 
		is the disjoint union of $D(p,q)$, $J(p,q)$ and $D(q,p)$.
	\item[(A3)] There exists a special point in the plane, which we call $p_\infty$,
		such that, for all distinct $p,q\in S$, 
		the curve $J(p,q)$ passes through $p_\infty$.% 
		\footnote{Usually the axiom tells that the stereographic projection 
		to the sphere of the curve $J(p,q)$ can be completed to a closed
		Jordan curve passing through the north pole. For us
		it will be more convenient to project from a different point
		and complete all curves within the plane to make them pass through $p_\infty$.}
	\item[(A4)] For each subset $S'$ of $S$ with 3 elements
		and each $p\in S'$, the abstract Voronoi region $\AVR(p,S')$
		is path connected. 
	\item[(A5)] For each subset $S'$ of $S$ with 3 elements
		we have $\RR^2=\bigcup_{p\in S'} \overline{\AVR(p,S')}$.
\end{enumerate}
For the rest of the discussion on abstract Voronoi diagrams, we assume
that these axioms are satisfied. Note that axioms (A4)-(A5) are not
the ones given in the definition of~\cite{KleinLN09} but, as they
show in their Theorem~15, they are equivalent. In this regard,
our definition is closer to the one given by Klein~\cite{Klein2014}.
Since we are going to work with very natural, non-pathological Voronoi diagrams,
any of the sets of axioms used in any of the other papers we have encountered
also works in our case.
Assuming these axioms, one can show that the abstract Voronoi diagram $\AVD(S)$
is a plane graph~\cite[Theorem~10]{KleinLN09}. 
This brings a natural concept of \DEF{abstract Voronoi vertex} 
and \DEF{abstract Voronoi edge} as those being vertices (of degree $\ge 3$)
and edges in the plane graph $\AVD(S)$.

Klein, Mehlhorn and Meiser provide a randomized incremental construction 
of abstract Voronoi diagrams. One has to be careful about what it means to 
compute an abstract Voronoi diagram, since it is not even clear how 
the input is specified. 
For their construction, they assume as primitive operation
that one can compute the abstract Voronoi diagram of any 
five abstract sites. 
The output is described by a plane graph $H$
and, for each vertex and each edge of $H$, a pointer to a vertex or an edge,
respectively, in the abstract Voronoi diagram for at most four abstract sites. Thus,
we tell that an edge $e$ of $H$ corresponds to some precise abstract edge $e'$ 
of $\AVD(S')$, where $|S'|\le 4$. Whether $\AVD(S')$ can be computed explicitly
or not, it depends on how the input bisectors can be manipulated.

Klein, Mehlhorn and Meiser consider a special case, 
which is the one we will be using, 
where the basic operation requires the abstract Voronoi diagram
of only four sites. 
(This particular case is not discussed by Klein, Langetepe and Nilforoushan~\cite{KleinLN09}, 
but they discuss the general case.)

\begin{theorem}[Klein, Mehlhorn and Meiser~\cite{KleinMM93}]
\label{thm:AVD}
	Assume that we have an admissible system of abstract bisectors for 
	a set $S$ of $m$ sites.
	The abstract Voronoi diagram of $S$ can be computed
	in $O(m\log m)$ expected time using an expected number of $O(m\log m)$ 
	elementary operations.
	If the abstract Voronoi diagram of any three sites contains 
	at most one abstract Voronoi vertex, besides the special point $p_\infty$,
	then an elementary operation is the computation of 
	an abstract Voronoi diagram for four sites.
\end{theorem}

%%%%%%%%%%%%%%%%%%%%%%%%%%%%%%%%%%%%%%%%%%%%%%%%%%%%%%%%%%%%%%%%%%%%%%%%%%%%%%%%%%%%%%%%%%%%
\section{Voronoi diagrams in planar graphs}
\label{sec:GVD}

We will need additively weighted Voronoi diagrams in \emph{plane} graphs.
We first define Voronoi diagrams for arbitrary graphs. 
Then we discuss a representation using the dual graphs
that works only for plane graphs and discuss some folklore properties.
See for example the papers of Marx and Pilipczuk~\cite{MarxP15} 
or Colin de Verdi\`ere~\cite{Verdiere10} for similar intuition. 
The dual representation is the key to be able to
use the machinery of abstract Voronoi diagrams as a black box. 

\subsection{Arbitrary graphs}
Let $G$ be an arbitrary graph, not necessarily planar, with no negative cycles.
A \DEF{site} $s$ is a pair $(v_s,w_s)$, where $v_s\in V(G)$ is
its \DEF{location}, and $w_s\in \RR$ is its \DEF{weight}, possibly negative.
With a slight abuse of notation, we will use $s$ instead of $v_s$ 
as the vertex. For example, for a site $s$
we will write $s\in V(G)$ instead of $v_s\in V(G)$ and 
$d_G(s,x)$ instead $d_G(v_s,x)$.

Let $S$ be a set of sites in $G$.
For each $s\in S$, its \DEF{graphic Voronoi region}, denoted $\cell_G(s,S)$, is defined by
\[
	\cell_G(s,S) ~=~ \{ x\in V(G)\mid \forall t\in S\setminus\{s\}:~
		w_s+d_G(s,x) \le w_t + d_G(t,x) \}.
\]
See Figure~\ref{fig:Voronoi} for an example. Note that we are using the distance from the sites to the vertices
to define the graphic Voronoi cells. For directed graphs, using the reverse distance from the vertex to the sites
 would define different graphic regions (in general). However, this is equivalent to use the reversed graph $G^R$ of $G$.

Even assuming that all distances in $G$ are distinct,
we may have $w_s  + d_G(s,x) = w_t + d_G(t,x)$ for some vertex $x$. 
Also, some Voronoi cells may be empty.
In our case, we will only deal with cases where these two things
cannot happen.
We say that the set $S$ of sites
is \DEF{generic} when, for each $x\in V(G)$ and for each distinct $s,t\in S$, 
we have $w_s  + d_G(s,x) \neq w_t + d_G(t,x)$.
The set $S$ is \DEF{independent} when each Voronoi cell is nonempty.
It is easy to see that, if $S$ is a generic, independent set of sites,
then $s\in \cell_G(s,S)$
and each vertex $x$ of $V(G)$ belongs
to precisely one graphic Voronoi cell $\cell_G(s,S)$ over all $s\in S$.

\begin{figure}
	\centering
	\includegraphics[page=6,scale=.9]{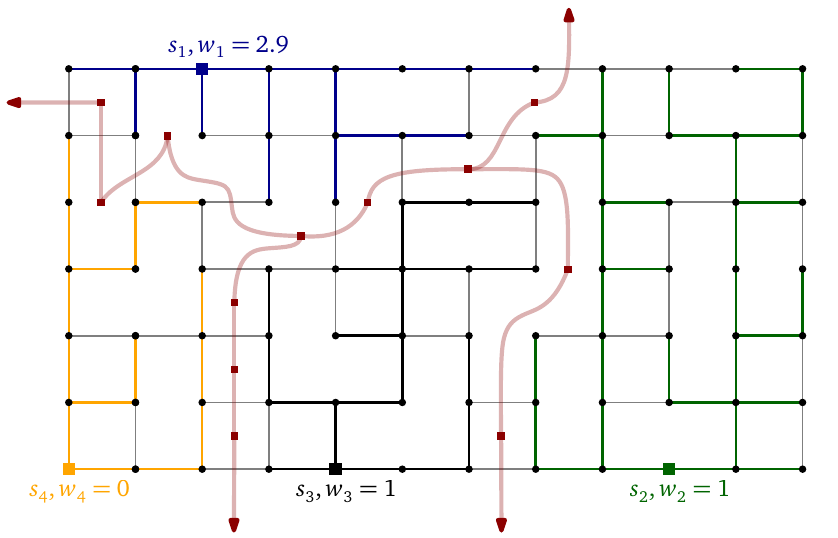}
	\caption{A graphic Voronoi diagram for four sites. The edges are undirected and have unit weight.
		An abstract Voronoi edge is marked with thicker pen.}
	\label{fig:Voronoi}
\end{figure}

The \DEF{graphic Voronoi diagram} of $S$ (in G) is the collection of graphic Voronoi regions:
\[
	\GVD_G(S) ~=~ \{ \cell_G(s,S)\mid s\in S\}.
\]
The following property is standard.

\begin{lemma}
\label{le:star}
	Let $S$ be a generic, independent set of sites.
	Then for each $s\in S$ the following hold:
	\begin{itemize}
		\item For each $x$ in $\cell_G(s,S)$, the shortest path
			from $s$ to $x$ is contained in $\cell_G(s,S)$.
		\item $\cell_G(s,S)$ induces a connected subgraph of $G$.
	\end{itemize}
\end{lemma}
\begin{proof}
	Let $x$ be a vertex of $\cell_G(s,S)$ and let $P(s,x)$
	be the shortest path in $G$ from $s$ to $x$.
	Assume, for the sake of reaching a contradiction,
	that some vertex $y$ on $P(s,x)$ is 
	contained in some other Voronoi cell $\cell_G(t,S)$, where $t\neq s$.
	Because of uniqueness of shortest paths, this means that $d_G(t,y)<d_G(s,y)$. 
	However, this implies that
	\[
		d_G(t,x)\le d_G(t,y)+d_G(y,x) < d_G(s,y)+d_G(y,x) = d_G(s,x),
	\]
	where in the last equality we have used that
	$y$ lies in the shortest path $P(s,x)$.
	The obtained inequality $d_G(t,x)<d_G(s,x)$ contradicts the 
	property that $x\in \cell_G(s,S)$. This proves the first item.
	
	To show the second item, 
	note that the subgraph of $G$ induced $\cell_G(s,S)$ contains (shortest) 
	paths from $s$ to all vertices of $\cell_G(s,S)$ because of the previous item.
\end{proof}

For each two sites $s$ and $t$, 
we define the \DEF{graphic dominance region} of $s$ over $t$ as
\begin{align*}
	\GD_G(s,t)~&=~ \cell_G(s,\{s,t\}) \\
			&=~ \{ x\in V(G)\mid 
		w_s+d_G(s,x) \le w_t+d_G(t,x) \}.
\end{align*}

\begin{lemma}
\label{le:intersection}
	For each $s\in S$ we have 
	$\cell_G(s,S)=\bigcap_{t\in S\setminus\{s \}} \GD_G(s,t)$.
\end{lemma}
\begin{proof}
	We note that  
	\begin{align*}
		\cell_G(s,S) ~&=~ \{ x\in V(G)\mid \forall t\in S\setminus \{s\}: 
						w_s+d_G(s,x) \le w_t+d_G(t,x) \}\\
					&=~ \bigcap_{t\in S\setminus \{s\}}\{ x\in V(G)\mid  
						w_s+d_G(s,x) \le w_t+d_G(t,x)\}\\
					&=~ \bigcap_{t\in S\setminus\{s \}} \GD_G(s,t).					
	\end{align*}
\end{proof}

%%%%%%%%%%%%%%%%%%%%%%%%%%%%%%%%%%%%%%%%%%%%%%%%%%%%%%%%%%%%%%%%%%%%%%%%%%%%%%%%%%%%%%%%%%%%%%%%%%%%%%%%%%%%%%%%%%%
\subsection{Plane graphs}
Now we will make use of graph duality to provide an alternative description
of additively weighted Voronoi diagrams in plane graphs. The aim is to define
Voronoi diagrams geometrically using bisectors, where
a bisector is just going to be a cycle in the dual graph.

Consider two sites $s$ and $t$ in $G$ and define
\[
	E_G(s,t) ~=~ \{ xy\in E(G) \mid x\in \GD_G(s,t),~ y\in \GD_G(t,s)\}. 
\]
Thus, we are taking the edges that
have each endpoint in a different graphic Voronoi region of $\GVD_G(\{s,t\})$.
We denote by $E^*_G(s,t)$ their dual edges.

\begin{lemma}
\label{le:bisector}
	Let $\{ s,t\}$ be a generic and independent set of sites.
	Then the edges of $E^*_G(s,t)$ define a cycle $\gamma$ in $G^*$.
	Moreover, if $s\in V_{\interior}(\gamma,G)$, then
	$V_{\interior}(\gamma,G)= \GD_G(s,t)$ and $V_{\exterior}(\gamma,G)=\GD_G(t,s)$.	
\end{lemma}
\begin{proof}
	Let $A^*$ be an arbitrary set of dual edges. It is well known
	that $A^*$ is the edge set of a cycle if and only if 
	$G-A$ has precisely two connected components. 
	Moreover, two faces $u^*$ and $v^*$ of $G^*$ 
	are in the same side of the cycle defined by $A^*$ 
	if and only if $u$ and $v$ are in the same connected component of $G-A$.
	See for example the proof in~\cite[Proposition 4.6.1]{Diestel05}
	or~\cite[Theorem 10.16]{BondyM08}.

	When $\{ s,t\}$ is generic and independent,
	we have $\GD_G(s,t)\neq \emptyset$, $\GD_G(t,s)\neq \emptyset$, and 
	$V(G)$ is the disjoint union of $\GD_G(s,t)$ and $\GD_G(t,s)$. 
	This means that $E_G(s,t)$ is the edge cut between
	$\GD_G(s,t)$ and its complement, $\GD_G(t,s)$.
	Moreover, by Lemma~\ref{le:star}, the subgraphs of $G$
	induced by $\GD_G(s,t)$ and by $\GD_G(t,s)$ are connected. 
	Therefore $G-E_G(s,t)$ has precisely two connected components, 
	and thus $E^*_G(s,t)$ is the edge set of a cycle $\gamma$ in $G^*$.

	Assume that $s\in V_{\interior}(\gamma,G)$.  
	Since $\GD_G(s,t)$ is the vertex set of the connected component
	of $G-E_G(s,t)$ that contains $s$, the faces 
	of $\{ u^*\mid u\in \GD_G(s,t)\}$ are in $\interior(\gamma)$
	and the faces $\{ v^*\mid v\in \GD_G(t,s)\}$ are in $\exterior(\gamma)$.
	Since a vertex $u$ of $G$ is the unique vertex of $G$ contained
	in the dual face $u^*$ of $G^*$, the result follows.
\end{proof}

\begin{figure}
	\centering
	\includegraphics[page=3,width=.49\textwidth]{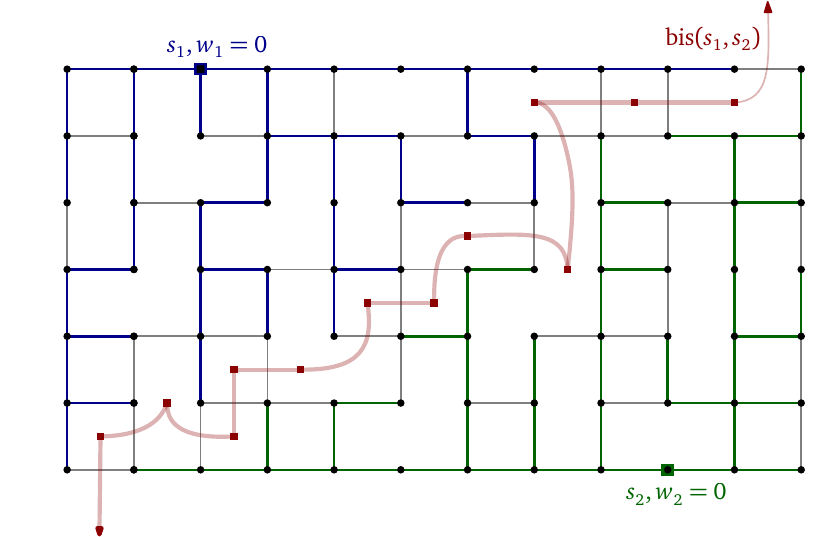}
	\includegraphics[page=5,width=.49\textwidth]{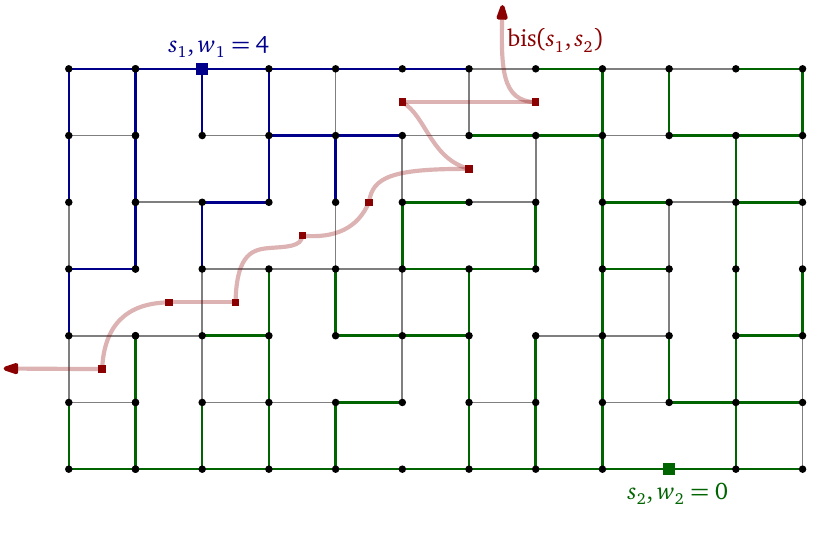}
	\caption{Two bisectors for sites placed at the same vertices but different weights. 
		The edges are undirected and have unit weight.}
	\label{fig:bisectors}
\end{figure}

When $s$ and $t$ are independent and generic,
we define the \DEF{bisector} of $s$ and $t$, denoted as $\bis_G(s,t)$,
as the curve in the plane defined by the cycle of $E^*_G(s,t)$,
as guaranteed in the previous lemma. 
See figure~\ref{fig:bisectors} for an example.
We also define $D_G(s,t)$ as the connected part of $\RR^2\setminus \bis_G(s,t)$ that contains $s$. 
We then have
\begin{equation}
	D_G(s,t)~=~ \left( \bigcup_{v\in \GD_G(s,t)} \overline{v^*}\right)^\circ .
	\label{eq:D_G}
\end{equation}
Here we have used the notation mentioned earlier: 
$\overline A$ and $A^\circ$ denote the closure and the interior of a set $A\subset \RR^2$, respectively.
Note that the pair $(\bis_G(s,t),D_G(s,t))$ is the type of pair
used to define abstract Voronoi diagrams.
From now on, whenever we talk about the abstract Voronoi
diagram of $G$, we refer to the abstract Voronoi
diagram defined by the system of bisectors 
$\{ (\bis_G(s,t),D_G(s,t))\mid s,t\in S, s\not= t\}$.

We have defined Voronoi regions of plane graphs 
in two different ways: using distances in
the primal graph $G$, called \emph{graphic} Voronoi regions, 
and using bisectors defined as curves in the plane, 
called \emph{abstract} Voronoi regions.
We next make sure that the definitions match, when restricted
to vertices of $G$. 

\begin{lemma}
\label{le:alternative}
	Let $G$ be a plane graph and 
	let $S$ be a generic, independent set of sites.
	Then, for each $s\in S$, we have $\cell_G(s,S)= V(G)\cap \AVR(s,S)$.
\end{lemma}
\begin{proof}
	Recall the definition
	\[
		\AVR(s,S) ~=~ \bigcap_{t\in S\setminus\{ s\}} D_G(s,t).
	\]
	Because of equation~\eqref{eq:D_G} we have
	\[ 
		D_G(s,t) ~=~ \left( \bigcup_{v\in \GD_G(s,t)} \overline{v^*}\right)^\circ ,
	\]
	and we obtain that 
	\begin{align*}
		\AVR(s,S) ~&=~ \bigcap_{t\in S\setminus\{ s\}}
					\left( \bigcup_{v\in \GD_G(s,t)} \overline{v^*}\right)^\circ
				~=~ \left( \bigcup_{v\in \bigcap_{t\in S\setminus\{ s\}} \GD_G(s,t)} \overline{v^*}\right)^\circ \\
				&=~ \left( \bigcup_{v\in \cell_G(s,S)} \overline{v^*}\right)^\circ,
	\end{align*}
	where in the last equality we used Lemma~\ref{le:intersection}.
	Since the only vertex of $V(G)$ contained 
	in the dual face $\overline{v^*}$ is precisely $v$, 
	and it lies in the interior of $v^*$,
	we get that 
	$V(G)\cap \AVR(s,S)=\cell(s,S)$.
\end{proof}

We cannot use the machinery of abstract Voronoi diagrams
for arbitrary sites because of axiom (A3). In our case
bisectors may not pass through a common ``infinity point" $p_\infty$.
Indeed, for arbitrary planar graphs we could have two bisectors that never intersect.
However, we can use it when all the sites are in the outer face of $G$. 
We next show this.

\begin{lemma}
\label{le:admissible}
	Let $G$ be a plane graph and 
	let $S$ be a generic, independent set of sites located
	in the outer face of $G$.
	Let $a_\infty$ be the vertex of $G^*$ dual to the outer face of $G$.
	Then the system of abstract bisectors 
	$\{ (\bis_G(s,t),D_G(s,t))\mid s,t\in S, s\not= t\}$
	is admissible, where $a_\infty$ plays the role of $p_\infty$ in axiom (A3).
\end{lemma}
\begin{proof}
	It is clear that the system of abstract bisectors 
	$\{ (\bis_G(s,t),D_G(s,t))\mid s,t\in S, s\not= t\}$
	satisfies axioms (A1) and (A2) of the definition.

	We next show the validity of axiom (A3).
	Consider any two sites $s$ and $t$ of $S$.
	Since $\cell_G(s,S)$ and $\cell_G(t,S)$ are nonempty, 
	also $\GD_G(s,t)$ and $\GD_G(t,s)$ are nonempty.
	Since $s$ and $t$ are located in the outer face of $G$
	the bisector $\bis_G(s,t)$ passes through $a_\infty$. 
	Indeed, the dual faces $s^*$ and $t^*$ have to be in different sides
	of the dual cycle $\bis_G(s,t)$ and, since $s$ and $t$ are on the outer face of $G$,
	that can happen only if $\bis_G(s,t)$ passes through $a_\infty$.
	Thus, if we take the geometric position of $a_\infty$ as $p_\infty$,
	all the bisecting curves pass through $p_\infty$ and axiom (A3) holds.
	
	For axiom (A4), consider any three sites $r,s,t$ of $S$ and let $S'=\{r,s,t\}$.
	As noted in the proof of Lemma~\ref{le:alternative},
	we have 
	\[
		\AVR(s,S')~=~ \left( \bigcup_{v\in \cell_G(s,S')} \overline{v^*}\right)^\circ.
	\]
	Since the vertices of $\cell_G(s,S')$ form a connected subgraph of $G$
	(Lemma~\ref{le:star}),
	the domains $\overline{v^*}$, when $v$ iterates over $\cell_G(s,S')$,
	are glued through the primal edges, and $\AVR(s,S')$ is path connected.
	This proves axiom (A4). 
	
	Axiom (A5) is shown similarly. Following the notation and the
	observations from the previous paragraph, we use that  
	\[
		\overline{\AVR(s,S')}~=~ \overline{\bigcup_{v\in \cell_G(s,S')} v^*}
	\]
	and that $V(G)= \cell_G(r,S')\cup\cell_G(s,S')\cup \cell_G(t,S')$,
	to conclude that
	\[
		\overline{\AVR(r,S')}\cup \overline{\AVR(s,S')}\cup \overline{\AVR(t,S')} 
		~=~ \overline{\bigcup_{v\in V(G)} v^*} ~=~ \RR^2.
	\] 
\end{proof}

The abstract Voronoi diagram $\AVD(S)$ is a plane graph,
and by construction it is contained in the dual graph $G^*$.
An abstract Voronoi vertex corresponds
to a vertex in the dual graph $G^*$.
An abstract Voronoi edge corresponds to
a path in the dual graph $G^*$. More precisely,
any abstract Voronoi edge corresponds
to a portion of a bisector $\bis_G(s,t)$ whose endpoints are vertices of $G^*$. 

We further have the following observation regarding the structure of abstract Voronoi diagrams.
\begin{lemma}
\label{le:1vertex}
	The abstract Voronoi diagram of any $3$ sites in the outer face
	of $G$ has at most one vertex, besides $a_\infty$.
\end{lemma}
\begin{proof}
	Assume that $S$ is the set of 3 sites.
	Since each site $s\in S$ is in the outer face, 
	the abstract Voronoi diagram $\AVD(s,S)$ contains the dual face $s^*$,
	which is incident to $a_\infty$. It follows that 
	all faces have a common vertex in $a_\infty$.
	Since a plane graph with $3$ faces 
	can have at most $2$ vertices of degree at least $3$, the result follows.	
\end{proof}

%%%%%%%%%%%%%%%%%%%%%%%%%%%%%%%%%%%%%%%%%%%%%%%%%%%%%%%%%%%%%%%%%%%%%%%%%%%%%%%%%%%%%
\subsection{Dealing with holes}
Let $G$ be a plane graph and let $P$ be a connected subgraph of $G$, with the embedding inherited from $G$.
Consider the graphic Voronoi diagram in $P$ \emph{using the distances in $G$}.
Thus, for a set of weighted sites $S$ in $P$ and a site $s\in S$ we are interested in the vertex subsets
\[
	\cell_{P,G}(s,S) ~=~ \cell_G(s,S)\cap V(P) .	
\]
Strictly speaking, $\cell_{P,G}(s,S)$ is a graphic Voronoi region in the complete graph
with vertex set $V(P)$ and edge lengths defined by the distances in $G$. 
We also have the graphic Voronoi diagram $\{ \cell_{P,G}(s,S)\mid s\in S\}$.
However, this interpretation in the complete graph
will not be very useful for us because it does not use planarity.
We would like to represent these Voronoi diagrams using the dual graph of $P^*$.
In particular, we have to define bisectors using the graph $P^*$.

\begin{figure}
	\centering
	\includegraphics[page=5,width=.9\textwidth]{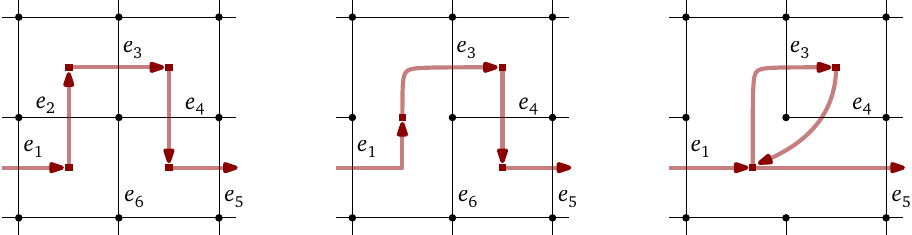}
	\caption{Transforming a cycle in $G^*$ into a non-crossing closed walk in $P^*$. 
			The portion $e_1^*,\dots, e^*_5$ of a cycle in $G^*$ (left) becomes $e_1^*,e_3^*,e_4^*, e^*_5$ with the deletion of $e_2$ (center)
			and keeps being $e_1^*,e_3^*,e_4^*, e^*_5$ with the deletion of $e_6$ (right). }
	\label{fig:holes1}
\end{figure}

\begin{lemma}
\label{le:withholes1}
	Given two sites $s$ and $t$, there is a non-crossing closed walk $\gamma$ in $P^*$ such that
	$\cell_{P,G}(s,\{s,t\})$ and $\cell_{P,G}(t,\{s,t\})$ are precisely 
	$V_{\interior}(\gamma,P)$ and $V_{\exterior}(\gamma,P)$ .
	Moreover, $\gamma$ is obtained from $\bis_G(s,t)$ by deleting the 
	edges of $(E(G)\setminus E(P))^*$ from the sequence of edges defining $\bis_G(s,t)$ .
\end{lemma}
\begin{proof}
	Let $e^*_1,\dots,e^*_k$ be the sequence of edges of $G^*$ that define $\bis_G(s,t)$.
	If in this sequence we delete all appearances of $e^*$ for $e\in E(G)\setminus E(P)$,
	then we obtain a subsequence $(e'_1)^*,\dots ,(e'_\ell)^*$ that defines a closed walk $\gamma$
	in $P^*$. See Figure~\ref{fig:holes1} for a small example and Figure~\ref{fig:holes2}
	for a larger example. The resulting closed walk is non-crossing, as can be seen by induction
	on the number of deleted edges. Indeed, if a plane graph $H'$ is obtained from a plane
	graph $H$ by deleting an edge $e$, then $(H')^*$ is obtained from $H^*$ by contracting $e^*$.
	Any non-crossing walk in $H^*$ remains non-crossing when contracting the edge $e^*\in E(H^*)$,
	and the interior of the walk contains exactly the same subset of the vertices of $H'$.
	Thus, it also follows by induction, that the vertices of $V(P)$ in the interior of $\bis_G(s,t)$
	remain in the interior during the contractions of the edges $e^*$ for $e\in E(G)\setminus E(P)$, 
	and therefore $\cell_{P,G}(s,S) ~=~ \cell_G(s,S)\cap V(P) = V_{\interior}(\gamma,P)\cap V(P)= V_{\interior}(\gamma,P)$.
	The same argument works for $V_{\exterior}(\gamma,P)$
	
	Note that our description of the transformation from $\bis_G(s,t)$ to $\gamma$
	using dual edges is simpler than a description using dual vertices. 
	This is so because the relevant faces may also change with deletions 
	of edges that are not crossed by $\bis_G(s,t)$.	

	The assumption that $P$ is connected is needed. Otherwise $P$ has faces that are not simply-connected,
	and closed walks of $G^*$ may become empty in $P^*$ because they do not cross any edge of $P$.
	Also, when $P$ has multiple components, there are curves that intersect the same edges of $P$ in the same order,
	but contain a different set of connected components in their interior. Thus, additional information
	beyond the edges of $P^*$ would be needed to encode the curves.
\end{proof}

\begin{figure}
	\centering
	\includegraphics[page=3,width=.49\textwidth]{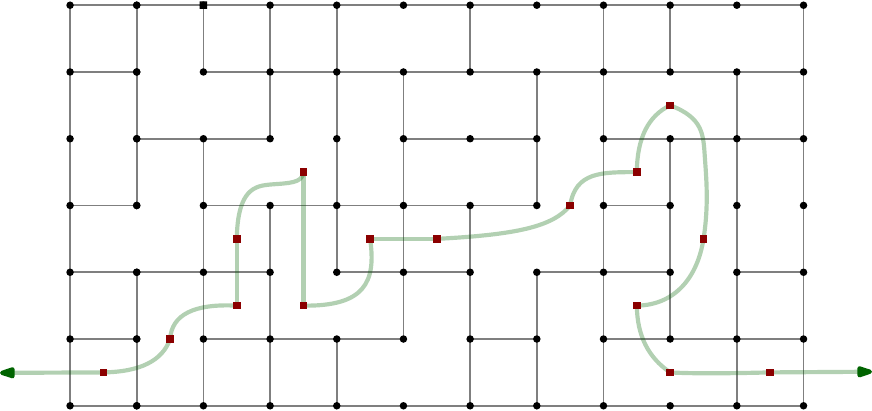}
	\includegraphics[page=4,width=.49\textwidth]{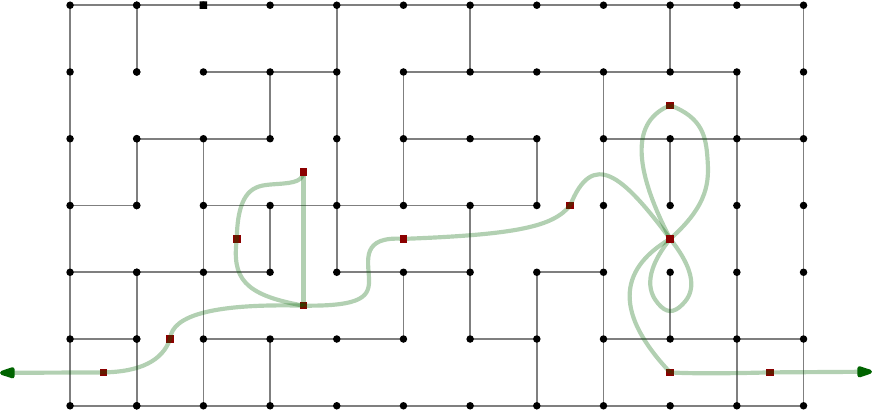}
	\caption{Transforming a cycle in $G^*$ into a non-crossing closed walk in $P^*$.}
	\label{fig:holes2}
\end{figure}

We use $\bis_{P,G}(s,t)$ for the non-crossing closed $\gamma$ in $P^*$ defined by Lemma~\ref{le:withholes1}.
To use abstract Voronoi diagrams we have the following technical problem:
in general, the curve $\bis_{P,G}(s,t)$ is not simple. 
We can work around this symbolically, as follows.
Combinatorially, we keep encoding the bisector as a closed walk in the dual graph $P^*$.
However, the geometric curve associated with a description goes out
of the dual graph to become simple. For each two consecutive edges $aa'$ and $a'a''$
of each such closed walk, we always make a small shortcut in a small
neighborhood of $a'$ that avoids $a'$. 
For example, we can reroute the arcs along small concentric circles,
where we use a larger radius when the distance along the face is smaller.
See Figure~\ref{fig:transformation} for an example.
There are different ways to do  this rerouting.
In any case, the algorithm of Theorem~\ref{thm:AVD} to build the abstract Voronoi diagram never uses coordinates.
In such a way we obtain true geometric simple curves associated to each such bisector.

The transformation is not made for the outer face. 
Indeed, to use the technology of abstract Voronoi diagrams,
we need that all the bisectors pass through a common point $p_\infty$, which is $a_\infty$.
Thus, we do not want to make any rerouting at the outer face. This is not a problem
if each bisector $\bis_{P,G}(s,t)$ passes exactly once through the vertex $a_\infty$.
If $G$ and $P$ have the same outer face, then $\bis_{P,G}(s,t)$ only
passes once through $a_\infty$. Thus, we will restrict attention
to the case when $G$ and $P$ have the same outer face.

\begin{figure}
	\centering
	\includegraphics[page=6,scale=1.2]{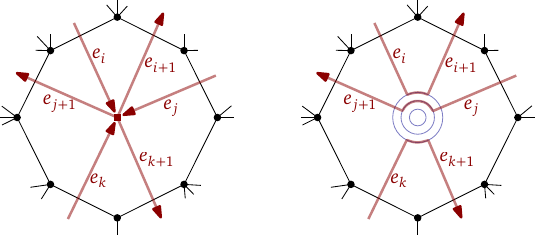}
	\caption{Transforming a closed walk in $P^*$ into a simple curve.}
	\label{fig:transformation}
\end{figure}

The rest of the presentation used for the case $G=P$ goes essentially unchanged.
However, note that Lemma~\ref{le:star} does not hold in this case.
The reason is that the shortest path from $s$ to $x$ may have edges outside $E(P)$.
An easier way to visualize things is to consider the creation of abstract Voronoi diagrams
in the graph $G^*$ and then consider the deletion of $E(G)\setminus E(P)$ in $G$ (and in $G^*$).
To summarize, we obtain the following.

\begin{lemma}
\label{le:admissible_holes}
	Let $G$ be a plane graph, let $P$ be a connected subgraph of $G$ such that $G$ and $P$
	have the same outer face.
	Let $S$ be a generic, independent set of sites located
	in the outer face of $P$.
	Let $a_\infty$ be the vertex of $G^*$ dual to the outer face of $G$.
	Then the system of abstract bisectors 
	$\{ (\bis_{P,G}(s,t),D_{P,G}(s,t))\mid s,t\in S, s\not= t\}$
	is admissible, where $a_\infty$ plays the role of $p_\infty$ in axiom (A3).

	The abstract Voronoi diagram of any $3$ sites in the outer face
	of $P$ has at most one vertex, besides $a_\infty$.
	
	For each $s\in S$, we have $\cell_{P,G}(s,S)= V(P)\cap \AVR(s,S)$.
\end{lemma}
\begin{proof}
	Because of Lemma~\ref{le:admissible}, the system of abstract bisectors 
	$\{ (\bis_G(s,t),D_G(s,t))\mid s,t\in S, s\not= t\}$
	is admissible. Because of Lemma~\ref{le:withholes1}, the system of abstract bisectors
	$\{ (\bis_{P,G}(s,t),D_{P,G}(s,t))\mid s,t\in S, s\not= t\}$ is obtained
	by deleting the edges of $(E(G)\setminus E(P))^*$ from the description of the bisectors,
	which amounts to contracting those edges in the dual graph.
	Consider a contraction of a dual edge and how it transforms the bisectors. 
	If we keep the bisectors as simple curves (not self-touching), as discussed above, 
	then the transformation of the bisectors during an edge contraction 
	can be done with a homeomorphism of the plane onto itself.
	Since the properties of being an admissible system of abstract bisectors are topological, 
	we obtain that $\{ (\bis_{P,G}(s,t),D_{P,G}(s,t))\mid s,t\in S, s\not= t\}$
	is admissible and $a_\infty$ plays the role of $p_\infty$. 
	
	The number of vertices does not change with the homeomorphism.
	Also, for any $s\in S$, the set $V(P)\cap \AVR(s,S)$ does not change during the homeomorphism,
	and therefore $\cell_{P,G}(s,S)= V(P)\cap \AVR(s,S)$.
\end{proof}

\paragraph{Remark.} 
Instead of using rerouting in the dual graph, another alternative 
is to use a variant of the line graph of the dual graph. 
The variant is designed to ensure that all the bisectors pass through $a_\infty$,
so that we can use abstract Voronoi diagrams.
Let us spell out an adapted construction of the graph, which we denote by $L_\infty$.
The vertex set of $L_\infty$ is $E(G)\cup \{ a_\infty\}$. Thus, each edge and the 
vertex $a_\infty$ corresponding to the outer face of $G$ are the vertices of $L_\infty$.
For each face $f$ of $G$ that is not the outer face, we put an edge in $L_\infty$ between each pair
of edges that appear in $f$.
For each edge $e$ on the outer face of $G$, we put an edge in $L_\infty$ between
$e$ and $a_\infty$. 
This finishes the description of $L_\infty$. The graph $L_\infty$ has a natural drawing inherited from the
embedding of $G$, that is not necessarily an embedding. ($L_\infty$ has large cliques when $G$ 
has large faces.) However, we can use a drawing of $L_\infty$ to represent the curves.
See Figure~\ref{fig:Linfty} for an example of (a drawing of) $L_\infty$.

Any non-crossing walk in $G^*$ that uses each edge at most once corresponds to
a cycle in $L_\infty$ because the portion of the walk inside a face $f$ of $G$ 
corresponds to a non-crossing matching inside $f$ between some edges on the boundary of $f$, 
and this matching is part of $L_\infty$. Intuitively, the edges of $L_\infty$ represent
shortcuts connecting edges of $G$ directly without passing through dual vertices.

\begin{figure}
	\centering
	\includegraphics[page=7,scale=1.15]{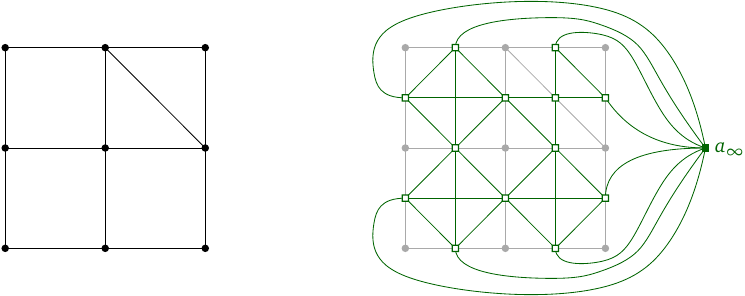}
	\caption{Left: a plane graph $G$. Right: a drawing of $L_\infty$ in green (and $G$ in gray).}
	\label{fig:Linfty}
\end{figure}

%%%%%%%%%%%%%%%%%%%%%%%%%%%%%%%%%%%%%%%%%%%%%%%%%%%%%%%%%%%%%%%%%%%%%%%%%%%%%%%%%%%%%
\section{Algorithmic aspects of Voronoi diagrams in planar graphs}
\label{sec:algorithms}
For the rest of this section, we assume that $G$ is a connected plane graph, 
$P$ is a connected subgraph of $G$, and 
the outer face of $P$ and $G$ coincide.
We use $r$ for the number of vertices in $P$.
Let $X$ be a set of $b$ vertices in the outer face of $P$.
We are interested in placing the sites at the vertices of $X$.
In this section we assume that the distances $d_G(\cdot,\cdot)$ 
from each vertex of $X$ to each vertex of $P$ are known and available.
We remark that the arcs of $G$ may have negative weights, but $G$ should not have negative cycles.

We next provide tools to manipulate portions of the bisectors and construct
Voronoi diagrams in planar graphs.

\begin{lemma}
\label{le:algbisector}
	For any two generic, independent sites $\{s,t\}$ placed at $X$
	we can compute $\bis_{P,G}(s,t)$ in $O(r)$ time.
\end{lemma}
\begin{proof}
	For each vertex $x\in V(P)$, we compare 
	$w_s+d_G(s,x)$ and $w_t+d_G(t,x)$ to decide whether 
	$x$ belongs to $\cell_{P,G}(s,\{ s,t \})$ or $\cell_{P,G}(t,\{ s,t \})$.
	Note that $w_s+d_G(s,x)\neq w_t+d_G(t,x)$ because we assume generic sites.
	The sets $\cell_{P,G}(s,\{ s,t \})$ and $\cell_{P,G}(t,\{ s,t \})$ are nonempty
	because we assume independent sites.
	Now we can mark the edges of $P$ with one endpoint in
	each of those sets and construct the closed walk $\bis_{P,G}(s,t)$ using the dual graph.
\end{proof}

\begin{lemma}
\label{le:2sites}
	Consider any two vertices $\{v_s,v_t\}\subseteq X$ as placements of sites.
	Consider the family of bisectors $\bis_{P,G}((v_s,w_s),(v_t,w_t))$
	as a function of the weights $w_s$ and $w_t$.
	There are at most $O(r)$ different bisectors.
	We can compute and store all the bisectors in $O(r^2)$ time
	such that, given two values $w_s$ and $w_t$,
	the corresponding representation of $\bis_{P,G}((v_s,w_s),(v_t,w_t))$
	is accessed in $O(\log r)$ time.
\end{lemma}
\begin{proof}
	From the definition it is clear that
	\[
		\bis_{P,G}((v_s,w_s),(v_t,w_t))= 
		\bis_{P,G} ((v_s,0),(v_t,w_t-w_s)).
	\]
	Thus, it is enough to consider the bisectors
	$\bis_{P,G}((v_s,0),(v_t,w))$ parameterized by $w\in \RR$. 
	Each bisector $\bis_G((v_s,0),(v_t,w))$ is a cycle in 
	the dual graph $G^*$
	and the cycles are nested: as $w$ increases, the graphic dominance
	region $\GD_G(s,t)$ monotonically grows
	and $D_G(s,t)$ also monotonically increases.
	The same happens with $\bis_{P,G}((v_s,0),(v_t,w))$: as $w$ increases,
	the bisectors $\bis_{P,G}((v_s,0),(v_t,w))$ are nested and the region on
	one side monotonically grows.
	Since any two different non-crossing closed walks $\bis_{P,G}((v_s,0),(v_t,w))$
	are nested and must differ by at least one vertex of $P$ that is enclosed,
	there are at most $O(r)$ different bisectors.

	For each vertex $x\in V(P)$, define the value $\eta_x=d_G(s,x)-d_G(t,x)$.
	The vertex $x$ is in $\GD_G(s,t)$ when $w<\eta_x$,
	in $\GD_G(t,s)$ when $w>\eta_x$,
	and we have a degenerate (non-generic) case when $w=\eta_x$.
	Thus, we can compute the values $\{ \eta_x \mid x\in V(P)\}$,
	sort them and store them sorted in a table. 
	For each $w$ between two consecutive values of $\{ \eta_x \mid x\in V(P)\}$
	we compute the bisector using Lemma~\ref{le:algbisector}
	and store it with its predecessor of $\{ \eta_x \mid x\in V(G)\}$.
	Given a query with shifts $w_s, w_t$,
	we use binary search in $O(\log r)$ time 
	for the value $w_t-w_s$ and locate the relevant bisector.
\end{proof}

As mentioned before, an abstract Voronoi vertex is just a vertex of $P^*$
and an abstract Voronoi edge is encoded in the dual graph $P^*$
by a tuple $(s,t,aa',bb')$, meaning that the edge is the portion of $\bis_{P,G}(s,t)$
starting with the dual edge $aa'$ and finishing with the the dual edge $bb'$
in some prescribed order, like for example the clockwise order
of $\bis_{P,G}(s,t)$.

\begin{lemma}
\label{le:3sites}
	Consider any three vertices $\{v_q,v_s,v_t\}\subseteq X$ as placements of sites.
	Consider the family of abstract Voronoi diagrams for the sites $q=(v_q,w_q)$,
	$s=(v_s,w_s)$, and $t=(v_t,w_t)$
	as a function of the weights $w_q$, $w_s$ and $w_t$.
	We can compute and store all those Voronoi diagrams in $O(r^2)$ time
	such that, given the values $w_q$, $w_s$ and $w_t$,
	the corresponding representation of the abstract Voronoi diagram
	of those $3$ sites is accessed in $O(\log r)$ time.
\end{lemma}
\begin{proof}
	We use Lemma~\ref{le:2sites} to compute 
	and store all the possible bisectors of each pair of vertices.
	This takes $O(r^2)$ time because we have $O(1)$ pairs of placements.

	Only the difference between weights of the sites is relevant.
	Thus, we can just assume that the weight $w_q$ is always $0$. 
	The relevant abstract Voronoi diagrams 
	can thus be parameterized by the plane $\RR^2$. The first coordinate is the
	weight $w_s$ and the second coordinate is the weight $w_t$.

	For each vertex $x\in V(P)$, we compute
	\[
		\eta_x^{qs} ~=~ d_G(q,x) -d_G(s,x), ~~~
		\eta_x^{qt} ~=~ d_G(q,x) -d_G(t,x), ~~~
		\eta_x^{st} ~=~ d_G(s,x) -d_G(t,x).
	\]
	Note that, once we fix the weights $w_s,w_t$ and $w_q=0$,
	the vertex $x\in V(P)$ belongs to $\cell_{P,G}(s,\{q,s,t\})$
	if and only if $w_s< \eta_x^{qs}$ and $w_t-w_s> \eta_x^{st}$.
	A similar statement holds for the other sites, $q$ and $t$.
	
	In the plane $(w_s,w_t)$ we consider the set of lines $L$
	that contains precisely the following lines
	\begin{align*}
		\forall x\in V(G):&~~ \ell_x^{qs}=\{ (w_s,w_t)\in \RR^2 \mid w_s=\eta_x^{qs} \}, \\
		\forall x\in V(G):&~~ \ell_x^{qt}=\{ (w_s,w_t)\in \RR^2 \mid w_t=\eta_x^{qt} \}, \\
		\forall x\in V(G):&~~ \ell_x^{st}=\{ (w_s,w_t)\in \RR^2 \mid w_t-w_s=\eta_x^{st}\}.
	\end{align*}
	Since $L$ has $O(r)$ lines,
	it breaks the plane $\RR^2$ into $O(r^2)$ cells, usually called
	the arrangement induced by $L$ and denoted by $\mathcal{A}(L)$.
	Such an arrangement can be computed in $O(r^2)$ time~\cite[Section 8.3]{bkos-08}.
	For each cell $c\in \mathcal{A}(L)$, the Voronoi diagram 
	defined by the sites $\{ (q,0), (s,w_s), (t,w_t)\}$ is the same for all $(w_s,w_t)\in c$.

	\begin{figure}
		\centering
		\includegraphics[page=1]{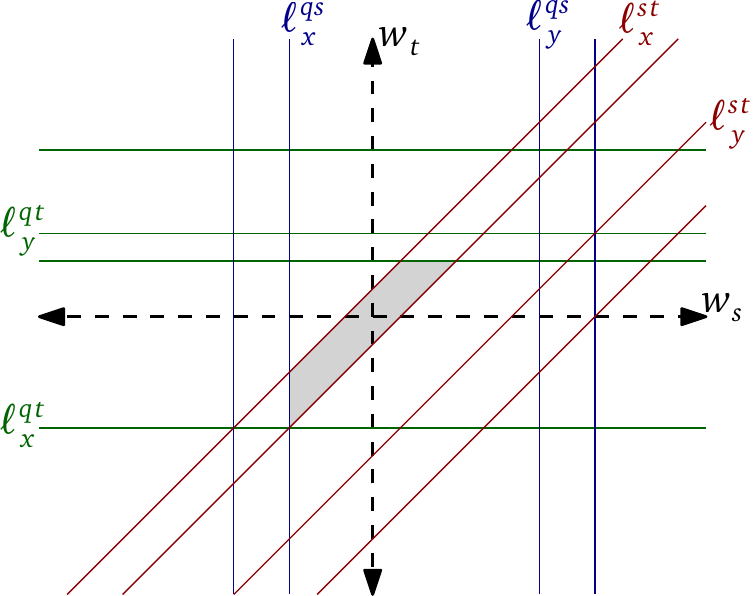}
		\caption{Example of lines $L$ defined by 4 vertices. 
			The cell of $\mathcal{A}(L)$ containing the origin is marked in gray.
			From the labels in th figure we can deduce that, in the Voronoi diagram of $q=(v_q,0)$,
			$s=(v_s,0)$ and $t=(v_t,0)$, the vertex $x$ belongs to $\cell_{P,G}(q,\{q,s,t\})$
			and the vertex $y$ belongs to $\cell_{P,G}(s,\{q,s,t\})$.}
		\label{fig:3sites}
	\end{figure}

	We can further preprocess $\mathcal{A}(L)$ for standard point location~\cite{Snoeyink04}.
	Thus, after $O(r^2)$ preprocessing,
	given a query point $(w_s,w_t)$, we can identify in $O(\log r)$ time
	the cell of $\mathcal{A}(L)$ that contains it. 

	In each cell $c$ of $\mathcal{A}(L)$ we store a description of the Voronoi diagram
	defined for weights on that cell.
	We can compute the relevant Voronoi diagram for each cell in $O(1)$ amortized time
	using a traversal of $\mathcal{A}(L)$. A simple way is as follows.
	Consider any line $\ell$ of $L$. Let us say that $\ell=\ell_x^{qs}\in L$;
	the other cases are similar. Let $\ell_\eps$ be a right shift
	of $\ell$ by an infinitesimal $\eps>0$.
	The value $w_s$ remains constantly equal to $\eta_x^{qs}+\eps$ as we walk along $\ell_\eps$,
	while the value $w_t$ changes.
	Consider the bisector $\bis_{P,G}((v_q,0),(v_s,w_s))$ and let $e_1,\dots,e_k$
	be the edges of $P$ that it crosses, as we walk from $a_\infty$ to $a_\infty$.
	Thus, the bisector is actually the non-crossing closed walk $e^*_1,\dots,e^*_k$
	in the dual graph $P^*$. For each such edge $e_i$, 
	we can compute a value $\zeta(e_i)$ such that
	$e_i$ is part of the abstract Voronoi edge of 
	$\{ (v_q,0), (v_s,w_s), (v_t,w_t)\}$ that separates the cell of $q$ and $s$
	if and only if $w_t> \zeta(e_i)$.
	Indeed, if $y_q$ is the endpoint of $e_i$ closer to $q$ and $y_s$
	the other endpoint,
	then $e_i$ is (part of) an abstract Voronoi edge of $\{ (v_q,0), (v_s,w_s), (v_t,w_t)\}$
	that separates the Voronoi cells of $q$ and $s$ 
	if and only if 
	\[ 
		w_q+d_G(q,y_q)< w_t +d_G(t,y_q) ~~\text{ and }~~
		w_s+d_G(s,y_s)< w_t +d_G(t,y_s).
	\]
	Using that $w_q=0$ and $w_s=\eta_x^{qs}+\eps$, this is equivalent to the condition
	\[
		w_t ~>~ \max\{ d_G(q,y_q)-d_G(t,y_q),~ d_G(s,y_s)-d_G(t,y_s)+ \eta_x^{qs}+\eps\} ~=:~ \zeta(e_i).
	\]
	Because of planarity, the values $\zeta(e_1),\dots,\zeta(e_k)$ are either
	monotonically increasing or decreasing. Indeed, the cell for $t$ can only grow
	when $w_t$ increases and the cell of $t$ has to take 
	always a contiguous part of the bisector $\bis_{P,G}((v_q,0),(v_s,w_s))$,
	as otherwise the Voronoi diagram of $\{ (q,0), (s,w_s), (t,w_t)\}$ would have at least $2$ vertices, besides $a_\infty$. 
	Therefore, the values $\zeta(e_1),\dots,\zeta(e_k)$
	are obtained already sorted. As we walk along $\ell_\eps$,
	we can identify the last edge $e_i$ such that $\zeta(e_i)< w_t$
	and identify the precise portion of $\bis_{P,G}((q,0),(s,w_s))$ that is in 
	the Voronoi diagram of $\{ (q,0), (s,w_s), (t,w_t)\}$.

	Repeating this procedure for each line $\ell^{qs}_x\in L$, 
	with two infinitesimal shifts per line, one on each side, 
	we can figure out in $O(1)$ amortized time per cell
	the portion of $\bis_{P,G}(q,s)$ in the abstract Voronoi diagram
	for each cell of $\mathcal{A}(L)$ bounded by one of those lines.
	If a cell is not bounded by a line $\ell^{qs}_x$ for some $x$,
	we figure out this information from a neighbour cell.
	A similar approach for the lines $\ell^{qt}_x\in L$ and $\ell^{st}_x\in L$
	determines the portions of $\bis_{P,G}(q,t)$ and $\bis_{P,G}(s,t)$, respectively.
	Thus, we obtain the abstract Voronoi diagrams for all cells $c\in \mathcal{A}(L)$ 
	in $O(1)$ amortized time per cell.
\end{proof}

Recall that $b$ is the cardinality of $X$.

\begin{lemma}
\label{le:4sites}
	There is a data structure with the following properties.
	The preprocessing time is $O(b^3 r^2)$.
 	For any generic, independent set $S$ of $4$ sites placed on $X$, 
	the abstract Voronoi diagram $\AVD(S)$ can be computed in 
	$O(\log r)$ time.
	The output is given combinatorially as a collection of
	abstract Voronoi vertices and edges encoded in the dual graph $P^*$.
\end{lemma}
\begin{proof}
	First, we make a table $T_X[\cdot]$ such that, for $u\in X$,
	$T_X[u]$ is the rank of $u$ when walking along the boundary of the outer face of $P$ and,
	for $u\notin X$, we have $T_X[u]$ undefined.
	Thus, given 3 vertices of $X$ we can deduce their circular ordering
	along the boundary of the outer face of $P$ in $O(1)$ time.

	We use Lemma~\ref{le:2sites} to compute 
	and store all the possible bisectors.
	Since there are $b^2$ different possible locations for the sites, 
	for each pair of locations there are $O(r)$ different bisectors, 
	and for each bisector we spend $O(r)$
	space and preprocessing time, we have spent a total of $O(b^2 r^2)$ time.

	For each bisector $\beta$, 
	we preprocess it to quickly figure out the circular order of its (dual) edges: 
	given two edges $aa'$ and $bb'$ on $\beta$,
	is the clockwise order along $\beta$ given by $aa',bb',a_\infty$ or by $bb',aa',a_\infty$?
	For each bisector $\beta$ we can make a table $T_\beta[\cdot]$ indexed
	by the edges such that $T_\beta[aa']$ is the position of $aa'$ along $\beta$,
	when we walk $\beta$ clockwise starting from $a_\infty$.
	We set $T_\beta[aa']$ to undefined when $aa'$ does not appear in $\beta$.
	Thus, given 2 edges of $\beta$, we can decide their relative order
	along $\beta$ in $O(1)$ time. 
	The time and space for this, over all bisectors, is also $O(b^2 r^2)$. 

	We make a table indexed by triples of vertices of $X$ and, for each triple,
	we use Lemma~\ref{le:3sites} and store in the table a pointer to the resulting data structure.
	We have $O(|X|^3)=O(b^3)$ choices for the vertices hosting the sites,
	and thus we spend $O(b^3 r^2)$ in the preprocessing step. 
	Given any three sites placed at $X$, we can get the abstract Voronoi diagram of those
	three sites in $O(\log r)$ time.
	This finishes the preprocessing.
	
	Assume that we are given a set $S$ of 4 sites placed at $X$ and 
	we want to compute its abstract Voronoi diagram.
	We recover the abstract Voronoi diagrams for each subset
	$\binom S3$ in $O(\log r)$ time, using the stored data. 

	If there are two sites $s,t\in S$ such that their bisector $\bis_{P,G}(s,t)$
	is in full in the Voronoi diagram of each subset $S'$ with 
	$|S'|=3$ and $\{s,t\}\subset S'\subset S$, then in the abstract
	Voronoi diagram of $S$ there is a region bounded only by $\bis_{P,G}(s,t)$.
	We can then compose that bisector and the abstract Voronoi diagram 
	of the other three sites
	to obtain the final Voronoi diagram. See the left of Figure~\ref{fig:4cases}.
	(It may be that we have more than one such ``isolated" abstract Voronoi region.)
	
	In the opposite case, in the abstract Voronoi diagram there is no abstract
	Voronoi region that is bounded by a unique bisector. The abstract Voronoi
	diagram restricted to the interior faces of $G$ is connected.
	The shape of such a Voronoi diagram can be only one of two, depending
	on which opposite sites share a common edge.
	See the center and right side of Figure~\ref{fig:4cases}.
	Let $p,q,s,t$ be the sites in clockwise order along the boundary of $G$.
	We can infer this order in $O(1)$ time through the table $T_X[\cdot]$.
	Assume, by renaming the sites if needed, that $\bis_{P,G}(s,p)$ has $s$ in its interior.
	From  $\AVD(\{p,q,s\})$ we obtain the edge $aa'$ of  $\bis_{P,G}(s,p)$ incident
	to the vertex of $\AVD(\{p,q,s\})$, and from $\AVD(\{p,s,t\})$ we obtain the edge $bb'$ of  
	$\bis_{P,G}(s,p)$ incident to the vertex of $\AVD(\{p,s,t\})$.
	If $a=b$, then  $\AVD(\{p,q,s,t\})$ has a common vertex of degree $4$
	that is incident to four abstract Voronoi edges.	
	If $aa'=bb'$ or the the cyclic order of $a_\infty,bb',aa'$ along $\bis_{P,G}(p,s)$ is clockwise,
	then the tuple $(p,s,aa',bb')$ defines an abstract edge in the abstract Voronoi diagram of $S$.
	Otherwise, there is tuple $(q,t,cc',dd')$ for some edges $cc'$ and $dd'$ that can be obtained
	by exchanging the roles of $p,s$ with $q,t$.
	From this information and the abstract Voronoi diagrams of each three sites,
	we can construct the abstract Voronoi diagram of $S=\{ p,q,s,t\}$.
\end{proof}

\begin{figure}
	\centering
	\includegraphics[width=\textwidth]{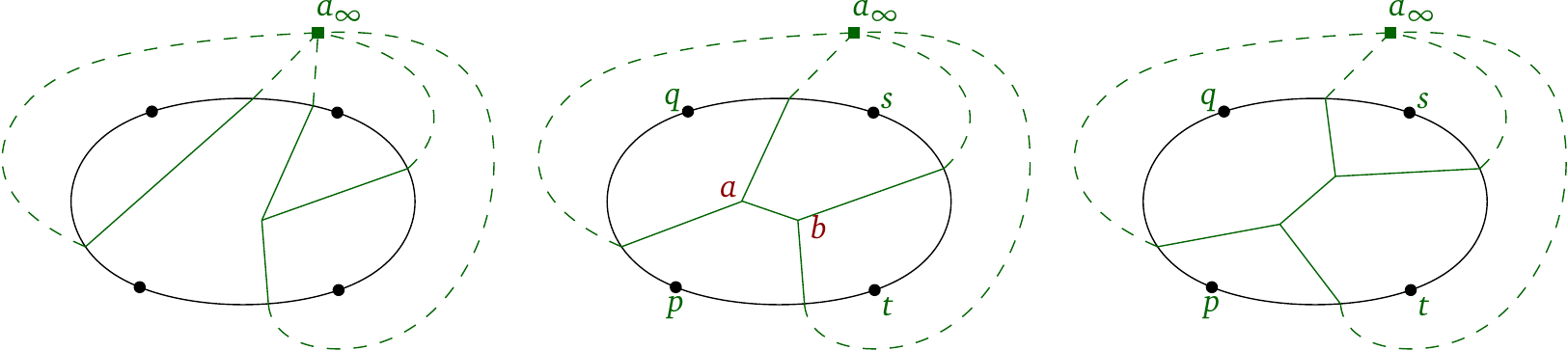}
	\caption{Left: an abstract Voronoi region is bounded by a single bisector.
		Center and right: possible configurations of the abstract Voronoi diagram
		of $4$ sites, when it is connected.}
	\label{fig:4cases}
\end{figure}

\begin{theorem}
\label{thm:abstract}
	Let $G$ be a plane graph and let $P$ be a connected subgraph of $G$
	with $r$ vertices such that $G$ and $P$ have the same outer face.
	Let $X$ be a set of $b$ vertices on the outer face of $P$.
	Assume that $G$ has no negative cycles
	and the distances $d_G(\cdot,\cdot)$ from each vertex of $X$ to each vertex of $P$ are available.
	There is a data structure with the following properties.
	The preprocessing time is $O(b^3 r^2)$.
 	For any generic and independent set $S$ of sites placed at $X$, 
	the abstract Voronoi diagram $\AVD_{P,G}(S)$ can be computed in 
	$\tilde O(b)$ expected time.
	The output is given combinatorially as a collection of
	abstract Voronoi vertices and edges encoded in the dual graph $P^*$.
\end{theorem}
\begin{proof}
	We apply the preprocessing of Lemma~\ref{le:4sites}.
	We spend $O(b^3r^2)$ time and, given any four sites placed on $X$, 
	we can compute its abstract Voronoi diagram in $O(\log r)$ time.

	Assume that we are given a set $S$ of $b$ sites placed at vertices of $X$. 
	Because of Lemma~\ref{le:admissible_holes} (see also Lemma~\ref{le:1vertex}),
	any three sites have a vertex in common, besides the one at $p_\infty$ (or $a_\infty$).
	According to Theorem~\ref{thm:AVD}, we can compute the 
	abstract Voronoi diagram using $O(|S|\log r)=O(b\log r)$ expected time
	and expected number of elementary operations, where an elementary operation 
	is the computation of an abstract Voronoi diagram of 4 sites. 
	Since each elementary operation takes $O(\log r)$ time because 
	of the data structure of Lemma~\ref{le:4sites}, the result follows.
\end{proof}

%%%%%%%%%%%%%%%%%%%%%%%%%%%%%%%%%%%%%%%%%%%%%%%%%%%%%%%%%%%%%%%%%%%%%%%%%%%%%%%%%%%%%%%%%%%%
\section{Data structure for planar graphs}
\label{sec:datastructure}
In this section we are going to use abstract Voronoi diagrams and
the data structures of Section~\ref{sec:dualcycle} to 
compute information about the distances from a fixed vertex
in a planar graph when the length of the edges incident to 
the fixed vertex are specified at query time.

Let $G$ be a plane graph with $n$ vertices and let $P$ be a connected subgraph of $G$
with $r$ vertices such that $G$ and $P$ have the same outer face.
Let $X$ be a set of $b$ vertices on the outer face of $P$.
Let $U$ be a subset of $V(P)$.
The graph $G$ may have arcs with negative edges, but it does not have any negative cycle.

For each subset $Y\subset X$, let $G^+(Y)$ be the graph obtained from $G$
by adding a new vertex $x_0$ and arcs $E_0(Y)=\{ \dart{x_0}{y} \mid y\in Y\}$.
See Figure~\ref{fig:G+}.
We want to preprocess $G$ and $P$ for different types of queries, as follows.
At preprocessing time, the lengths of the edges in $E_0(X)$ are undefined, unknown.
At query time we are given a subset $Y\subseteq X$ 
and the lengths $\lambda(\dart{x_0}{y})$ for the arcs $\dart{x_0}{y}$ of $E_0(Y)$. 
Using the notation introduced in Section~\ref{sec:preliminaries},
we are interested in the following information about the distances from the new vertex $x_0$:
\begin{align*}
	\diam(x_0,U,G^+(Y))~&=~\max \{d_{G^+(Y)}(x_0,u)\mid u \in U\},\\
	\adding(x_0,U,G^+(Y))~&=~\sum_{u\in U} d_{G^+(Y)}(x_0,u),\\
	\counting(x_0,U,G^+(Y),\delta) ~&=~ |\{ u \in U \mid d_{G^+(Y)}(x_0,u)\le \delta \}|.
\end{align*}
Note that we are only using the distances to the subset $U\subseteq V(P)$.

The set of vertices $Y$ and the lengths $\lambda(\dart{x_0}{y})$, where $y\in Y$,
will be given so that they satisfy the following condition:
\begin{equation}
\label{eq:condition}
	\forall y,y'\in Y,~ y\neq y': ~~~\lambda(\dart{x_0}{y}) < \lambda(\dart{x_0}{y'}) + d_{G}(y',y).
\end{equation}
This condition implies that, for all $y\in Y$, there is a unique shortest path
from $x_0$ to $y\in Y$ and this shortest path is just the arc $\dart{x_0}{y}$.
This condition is important in our scenario to ensure that, when using the vertices of $Y$
as sites with weights $\lambda(\dart{x_0}{y})$, the sites are generic and independent.

\begin{figure}
	\centering
	\includegraphics[page=1]{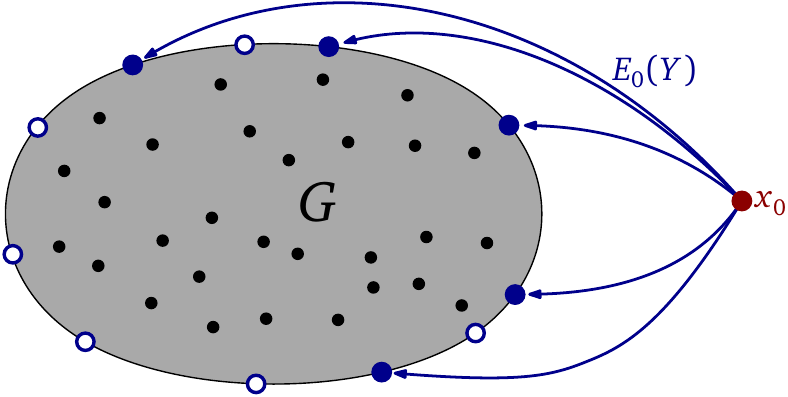}
	\caption{The graph $G^+(Y)$. The graph $G$ is represented by the gray region,
		vertices of $Y$ are marked with filled in dots, and vertices of $X\setminus Y$
		are represented with void dots.}
	\label{fig:G+}
\end{figure}

\begin{theorem}
\label{thm:perpiece_sum}
	Assume that $G$ is a weighted plane graph with $n$ vertices and no negative cycles.
	Let $P$ be a subgraph of $G$ with $r$ vertices 
	such that $G$ and $P$ have the same outer face.
	Let $X$ be a set of $b$ vertices on the outer face of $P$
	and let $U$ be a subset of $V(P)$.
	After $\tilde O(n + b^3 r^2)$ preprocessing time, we can handle
	the following queries in $\tilde O(b)$ expected time:
	given a subset of vertices $Y\subset X$ and weights for the darts
	$\lambda(\dart{x_0}{y})$, $y\in Y$, that satisfy the condition~\eqref{eq:condition},
	return $\adding(x_0,U,G^+(Y))$.
\end{theorem}
\begin{proof}
	We compute and store the distances $d_G(x,v)$ for all $x\in X$ and $v\in V(P)$.
	This can be done in $\tilde O(n+br)$ time, as follows.
	First we compute a single-source shortest-path tree 
	in $\tilde O(n)$ time~\cite{fr-pgnwe-06,KleinMW10,MozesW10}. 
	With this we have a potential function in $G$
	and for the next distances we can assume non-negative weights.
	Then we use that all the vertices of $X$ are incident to the outer face of $G$.
	Using~\cite{CabelloCE13,k-msspp-05} we obtain
	in $\tilde O(|V(P)|+ |X|\cdot |V(P)|)= \tilde O(n+br)$ time the distances $d_G(x,v)$, 
	for all $x\in X$ and $v\in V(P)$.
	
	We preprocess the pair of graphs $G$ and $P$ as described in Theorem~\ref{thm:abstract}.
	This takes $O(b^3r^2)$ time because we already have all the required distances.
	
	For each vertex $x\in X$ we proceed as follows. 
	We compute all the bisectors of the type $\bis_{P,G}((x,\cdot),(x',\cdot))$ for all $x'\in X$. 
	Let $\Pi_x$ be the resulting family of curves.
	Then, we preprocess $P$ with respect to $\Pi_x$ 
	as explained in Theorem~\ref{thm:sumweights}.
	More precisely, we use Theorem~\ref{thm:sumweights} for the following two vertex-weight functions:
	\begin{align*}
		\omega_x(v) ~=~ \begin{cases}
						d_G(x,v), &\text{if $v\in U$,}\\
						0, &\text{otherwise.}
			\end{cases}~~~~
		\omega'_x(v)~=~ \begin{cases}
						1, &\text{if $v\in U$,}\\
						0, &\text{otherwise.}
			\end{cases}
	\end{align*}
	We denote by $\sigma_x(\cdot)$ and $\sigma'_x(\cdot)$ the corresponding sums of weights.
	For example, $\sigma'_x(U')=\sum_{v\in U'} \omega'_x(v)$ for all $U'\subseteq V(P)$.
	This finishes the description of the preprocessing.

	Let us analyze the running time for the last step of the preprocessing.
	For each two vertices $x,x' \in X$, we spend $O(r^2)$ time 
	to compute the bisectors $\bis_{P,G}((x,\cdot),(x',\cdot))$ because of Lemma~\ref{le:2sites}. 
	It follows that $\Pi_x$ is a family of walks in $P^*$ with $O(br^2)$ dual edges, 
	counted with multiplicity.
	The preprocessing of Theorem~\ref{thm:sumweights}
	is $O(r+||\Pi_x||)=O(r+br^2)=O(br^2)$ per vertex $x\in X$, 
	where $||\Pi_x||$ denotes the number of edges in $\Pi_x$.
	Thus, over all $x\in X$, we spend $O(b^2r^2)$ time.
	
	Consider now a query specified by a subset $Y\subset X$
	and the edge weights $\lambda( \dart{x_0}{y})$, $y\in Y$, 
	that satisfy the condition~\eqref{eq:condition}. 
	For each $y\in Y$, define the site $s_y=(y,\lambda( \dart{x_0}{y}))$.
	Because of condition~\eqref{eq:condition}, 
	the set of sites $S_Y=\{ s_y\mid  y\in Y\}$ is independent and generic.
	Using the data structure of Theorem~\ref{thm:abstract},
	we compute the weighted Voronoi diagram for the sites $S_Y$.	
	Thus, we obtain the abstract Voronoi diagram in 
	$\tilde O(|Y|) = \tilde O(b)$ expected time.
	For each $y\in Y$, let $\gamma_y$ be the closed walk in the dual graph $P^*$ that
	defines the boundary of $\AVR(s_{y(v)},S)$.
	
	For each vertex $v\in V(P)$ there is precisely one vertex $y(v) \in Y$ such that
	$\AVR(s_{y(v)},S)$ contains $v$. Moreover, because of the definition of (graphic)
	Voronoi diagrams, we have 
	\[
		d_{G^+(Y)}(x_0,v) ~=~ \lambda(\dart{x_0}{y(v)}) + d_G(y(v),v).
	\]

	Note that
	\begin{align*}
		\adding(x_0,U,G^+(Y))~&=~\sum_{u\in U} d_{G^+(Y)}(x_0,u)\\
			&=~ \sum_{u\in U} \left( \lambda(\dart{x_0}{y(v)}) + d_G(y(v),v) \right)\\
			&=~ \sum_{y\in Y} ~\sum_{u\in U\text{ s.t.} y(u)=y } \left( \lambda(\dart{x_0}{y(v)}) + d_G(y(v),v) \right)\\
			&=~ \sum_{y\in Y} ~\sum_{u\in \AVR(s_y,S)\cap U } \left( \lambda(\dart{x_0}{y(v)}) + d_G(y,v) \right)\\			
			&=~ \sum_{y\in Y} \left( \lambda(\dart{x_0}{y(v)}) \cdot |\AVR(s_y,S)\cap U| 
							+ \sum_{u\in  \AVR(s_y,S)\cap U} d_G(y,v) \right)\\
			&=~ \sum_{y\in Y} \left( \lambda(\dart{x_0}{y(v)}) \cdot \sigma'_y(V_{\interior}(\gamma_y,P)) 
							+ \sigma_y (V_{\interior}(\gamma_y,P)) \right) \numberthis \label{eq:add}
	\end{align*}
	For each site $y\in Y$, we walk along $\gamma_y$, the boundary 
	of the abstract Voronoi region $\AVR(s_y,Y)$, and use the data structures
	of Theorem~\ref{thm:sumweights} for $\omega_y$ and $\omega'_y$ to collect the data
	\begin{align*}
		\forall y\in Y: ~~~\sigma_y(V_{\interior}(\gamma_y,P)) , \text{ and } \sigma'_y(V_{\interior}(\gamma_y,P)) .
	\end{align*}
	Here we are using that $y\in \AVR(s_y,S)$, and thus $y$ is in the interior of $\gamma_y$.
	For each $\gamma_y$ we spend $\tilde O(1)$ times the complexity of its description. 
	Over all $Y$, this takes	$\tilde O(|Y|)=\tilde O(b)$ time. 
	From this information we can compute $\adding(x_0,U,G^+(Y))$ using ~\eqref{eq:add},
	and the result follows.
\end{proof}

\begin{theorem}
\label{thm:perpiece_max}
	Consider the setting of Theorem~\ref{thm:perpiece_sum}.
	After $\tilde O(nb + b^3 r^2 + b^4)$ preprocessing time, we can handle
	the following queries in $\tilde O(b)$ expected time:
	given a subset of vertices $Y\subset X$, weights for the darts
	$\lambda(\dart{x_0}{y})$, $y\in Y$, that satisfy the condition~\eqref{eq:condition},
	return $\diam(x_0,U,G^+(Y))$.
\end{theorem}
\begin{proof}
	We use the same approach as in the proof of Theorem~\ref{thm:perpiece_sum}.
	We keep using the notation of that proof.
	The main difference is that we do not use the data structure of Theorem~\ref{thm:sumweights},
	but the data structure of Theorem~\ref{thm:maxweight}. We explain the details of this part.
	
	For each vertex $x\in X$ we proceed as follows. 
	Let $T_x$ be a shortest-path tree in $G$ from $x$. We do not compute $T_x$, but use it
	to argue correctness.
	Then, we use the data structure of Theorem~\ref{thm:maxweight}
	for $G$, $P$, the tree $T_x$, and the vertex-weights $\omega_x(\cdot)$.
	Let $\mu_x$ be the corresponding maximum function that the data structure returns.
	Thus, $\mu_x(U)=\max \{ \omega_x(u)\mid u\in U\}$.
	This finishes the description of the preprocessing.

	Let us analyze the running time for this step of the preprocessing.
	Like before, each $\Pi_x$ is computed in $O(br^2)$ time and has $O(br^2)$ dual edges, 
	counted with multiplicity.
	The preprocessing of Theorem~\ref{thm:maxweight} is
	$O(n+||\Pi_x||+b^3)= O(n+br^2+b^3)$ time for each $x\in X$. 
	Therefore, the total preprocessing used in this step
	is $O(nb+b^2r^2+b^4)$.
	
	Next, we note that each $\gamma_y$ is in $\widetilde\Xi(G,P,T_y)$ because
	of Lemma~\ref{le:star}. Therefore, we can obtain $\mu_y(V_{\interior}(\gamma_y,P))$
	in $\tilde O(1)$ times the complexity of the description of $\gamma_y$.
	Over all $y\in Y$, this takes	$\tilde O(|Y|)=\tilde O(b)$ time.

	With this data, the desired value is then obtained in $O(|Y|)=O(b)$ time using that
	\begin{align*}
		\diam(x_0,U,G^+(Y))~&=~\max \{ d_{G^+(Y)}(x_0,u) \mid u\in U\} \\
			&=~ \max_{y\in Y} ~\max_{u\in U\text{ s.t.} y(u)=y } \left( \lambda(\dart{x_0}{y(v)}) + d_G(y(v),v) \right)\\
			&=~ \max_{y\in Y} ~\left( \lambda(\dart{x_0}{y(v)}) ~+ \max_{u\in U\text{ s.t.} y(u)=y } d_G(y(v),v) \right)\\
			&=~ \max_{y\in Y} ~\left(\lambda(\dart{x_0}{y(v)}) ~+ \mu_y(V_{\interior}(\gamma_y,P))\right). \qedhere
	\end{align*}
\end{proof}

\begin{theorem}
\label{thm:perpiece_count}
	Consider the setting of Theorem~\ref{thm:perpiece_sum}.
	After $\tilde O(n + b^3 r^3)$ preprocessing time, we can handle
	the following queries in $\tilde O(b)$ expected time:
	given a subset of vertices $Y\subset X$, weights for the darts
	$\lambda(\dart{x_0}{y})$, $y\in Y$, that satisfy the condition~\eqref{eq:condition},
	and a real value $\delta$,
	return $\counting(x_0,U,G^+(Y),\delta)$.
\end{theorem}
\begin{proof}
	We use the same approach as in the proof of Theorem~\ref{thm:perpiece_sum} and keep using its notation.
	The main difference is that we do not use the data structure of Theorem~\ref{thm:sumweights},
	but the data structure of Corollary~\ref{cor:countweight} for the vertex-weights $\omega_x(\cdot)$.
	Let $\kappa_\le^x$ be the corresponding function.
	This means that, for each $x\in X$, we spend an extra factor $r$ in the preprocessing.
	Thus, for each $x$ we spend $O(b^2r^3)$ time, instead of $O(b^2r^2)$.
	Over all $x\in X$, this means that the preprocessing has an extra factor of $O(b^3r^3)$.
	
	The rest of the approach is the same. We just have to use that
	\begin{align*}
		\counting(x_0,U,G^+(Y),\delta)~&=~|\{ u \in U \mid d_{G^+(Y)}(x_0,u)\le \delta \}| \\
			&=~ \sum_{y\in Y} ~|\{ u \in  V_{\interior}(\gamma_y,P) \cap U \mid \lambda(\dart{x_0}{y(v)}) + d_G(y(v),v) \le \delta \} | \\
			&=~ \sum_{y\in Y} ~ \kappa_\le^y(V_{\interior}(\gamma_y,P ), \delta-\lambda(\dart{x_0}{y(v)}),
	\end{align*}
	and all values $\kappa_\le^y(V_{\interior}(\gamma_y,P ), \delta-\lambda(\dart{x_0}{y(v)})$ are recovered from the data
	structure of Corollary~\ref{cor:countweight} in $\tilde O(|Y|)=\tilde O(b)$ time.
\end{proof}

%%%%%%%%%%%%%%%%%%%%%%%%%%%%%%%%%%%%%%%%%%%%%%%%%%%%%%%%%%%%%%%%%%%%%%%%%%%%%%%%%5
\section{Diameter and Sum of Distances in Planar Graphs}
\label{sec:together}

The data structures of Theorems~\ref{thm:perpiece_sum}, \ref{thm:perpiece_max} and~\ref{thm:perpiece_count}
are going to be used for each piece of an $r$-division. Then, for each vertex of $G$ we are going to query them.
We first explain the precise concept of piece and division that we use, and then explain its use.

\paragraph{Divisions.}
The concept of $r$-division for planar graphs
was introduced by Frederickson~\cite{f-faspp-87},
and then refined and used by several authors;
see for example~\cite{Cabello12,goodrich95,KleinS98,ItalianoNSW11} for a sample.
For us it is most convenient to use the construction of 
Klein, Mozes and Sommer~\cite{KleinMS13}.
We first state the definitions carefully, almost verbatim from the work of Klein, Mozes and Sommer~\cite{KleinMS13}.

Let $G$ be a plane graph. 
A \DEF{piece}\footnote{They use the term ``region", 
which in our opinion is more suitable. However, we are using such a term 
for so many things in this paper that in our context we prefer to use some other term.} 
$P$ of $G$ is an edge-induced subgraph of $G$. In each
piece we assume the embedding inherited from $G$.
A \DEF{boundary vertex} of a piece $P$ is a vertex of $P$ that is incident to some
edge in $E(G)\setminus E(P)$. 
A \DEF{hole} of a piece $P$ is a face of $P$ that is not a face of $G$.
Note that each boundary vertex of a piece $P$ is incident to some hole of $P$.
An \DEF{$r$-division with a few holes} of $G$ is a collection $\{ P_1,\dots, P_k\}$
of pieces of $G$ such that 
\begin{itemize}
	\item there are $O(n/r)$ pieces, that is, $k=O(n/r)$;
	\item each edge of $G$ is in at least one piece;
	\item each piece has $O(r)$ vertices;
	\item each piece has $O(\sqrt{r})$ boundary vertices;
	\item each piece has $O(1)$ holes.
\end{itemize}

\begin{theorem}[Klein, Mozes and Sommer~\cite{KleinMS13}]
\label{thm:division}
	There is a linear-time algorithm that, 
	for any biconnected triangulated planar embedded graph $G$, outputs
	an $r$-division of $G$ with few holes.
\end{theorem}

In fact, we will only use that all pieces together have $O(n/r)$ holes.
Thus, other decompositions proposed by other authors could also be used.
Note that we can assume that each piece is connected because we could replace
each piece by its connected components, and we would get a new $r$-division
with a few holes.

\paragraph{Work per piece.}
We now describe how to compute the relevant information within a fixed piece
and the information between a fixed piece and all vertices outside the piece.
The next result is sufficient for our purposes; 
better results can be obtained using additional tools~\cite{MozesNW14,MozesS12}.

\begin{lemma}
\label{le:withinpiece}
	Let $P$ be a piece of $G$ with $r$ vertices and $O(\sqrt{r})$ boundary vertices. 
	Let $U$ be a subset of vertices in $P$.
	In $\tilde O(nr^{1/2}+r^2)$ time we can compute for all vertices $v\in V(P)$
	the values $\diam(v,U,G)$, $\adding(v,U,G)$, and $\counting(v,U,G,\delta)$ (for a given $\delta\in \RR$).
\end{lemma}
\begin{proof}
	Let $\partial$ be the set of boundary vertices of $P$ in $G$.
	We compute shortest-path trees from each vertex $x\in \partial$ in $G$ 
	in near-linear time~\cite{fr-pgnwe-06,KleinMW10,MozesW10}. 
	This takes $|\partial|\cdot \tilde O(n)= \tilde O(nr^{1/2})$ time.

	We build a graph $\widetilde P$ by adding to $P$ arcs
	between each pair of vertices of $\partial$. The length of each
	new arc $\dart{x}{y}$ is set to $d_G(x,y)$.
	Standard arguments show that a distance between any two vertices of $P$ 
	is the same in $G$ and in $\widetilde P$.
	The graph $\widetilde P$ has $O(|E(P)|+|\partial|^2)=O(r)$ edges and $O(r)$
	vertices. We can compute all pairwise distances in 
	$\tilde O(|V(\widetilde P)|\cdot |E(\widetilde P)|)=\tilde O(r\cdot r)=\tilde O(r^2)$ time using
	standard approaches. (Since $\widetilde P$ may have negative weights, we
	may have to use a potential function.)
	
	From all the distances in $\widetilde P$, that are also distances in $G$,
	we can compute the desired values directly.
\end{proof}

\begin{lemma}
\label{le:outside_piece}
	Let $P$ be a piece of $G$ with $r$ vertices, $O(\sqrt{r})$ boundary vertices, and $h$ holes. 
	Let $U$ be a subset of vertices in $P$.
	\begin{itemize}
		\item In $\tilde O(nh + r^{7/2} + nr^{1/2} )$ expected time we can compute 
			the values $\adding(v,U,G)$ for all vertices $v\in V(G)\setminus V(P)$.
		\item In $\tilde O(nh + r^{7/2} + nr^{1/2} )$ expected time we can compute 
			the values $\diam(v,U,G)$ for all vertices $v\in V(G)\setminus V(P)$.
		\item In $\tilde O(nh + r^{9/2} + nr^{1/2} )$ expected time we can compute 
			the values $\counting(v,U,G,\delta)$ for a given $\delta$ and all vertices $v\in V(G)\setminus V(P)$.
	\end{itemize}
\end{lemma}
\begin{proof}
	Let $C_1,\dots, C_h$ be the facial walks of the holes of $P$. 
	For $i\in [h]$, let $A_i$ be the vertices of $G$ contained in the interior of the hole defined by $C_i$.
	Since each vertex of $V(G)\setminus V(P)$ is strictly contained in exactly one hole of $P$,
	the sets $A_1,\dots,A_h$ form a partition of $V(G)\setminus V(P)$.
	For each $i\in [h]$, we define the graph $G_i=G-A_i$ and  
	let $X_i$ be the set of boundary vertices that appear in $C_i$.
	See Figure~\ref{fig:perpiece} for an example.
	The sets $A_1,\dots,A_h$, $X_1,\dots,X_h$,
	and the graphs $G_1,\dots, G_h$ can be constructed in $O(nh)$ time.

	\begin{figure}
		\centering
		\includegraphics[page=1]{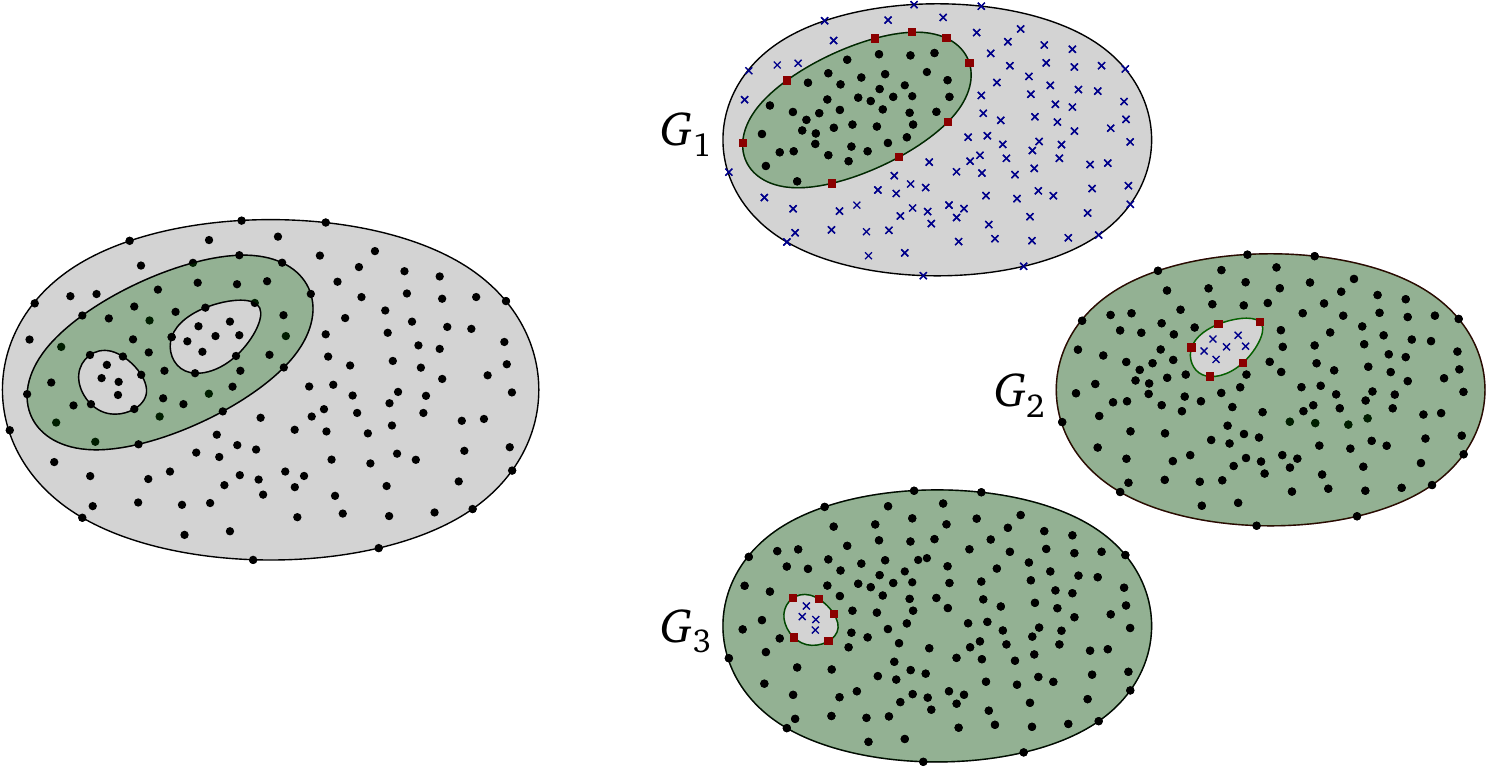}
		\caption{The graphs $G_i$ and the vertices $X_i$ in the proof of Lemma~\ref{le:outside_piece}.
			On the left we have a graph $G$ (in gray) and a piece $P$ (in light green) of $G$ with three holes.
			On the right we have, for $i\in [3]$, the graph $G_i$ (in light green), the set $X_i$ (red squares) and the set $A_i$ (blue crosses).}
		\label{fig:perpiece}
	\end{figure}

	We compute the distances in $G$ and in $G_i$ from and to all boundary vertices $X_i$.
	This can be done computing $4\cdot |X_i|= O(\sqrt{r})$ different shortest-path trees, each of
	them in $G$, $G_i$, or the reversed graphs $G^R$,  $G_i^R$.
	Since each single-source shortest path can be computed in
	$\tilde O(n)$ time~\cite{fr-pgnwe-06,KleinMW10,MozesW10},
	we spend in total $\tilde O(nr^{1/2})$ time. 	

	Consider any fixed index $i\in [h]$.
	For each $v\in A_i$, let $Y_i^v$ be the vertices $y$ of $X_i$ such that
	in the shortest path in $G$ from $v$ to $y$ the last arc is not contained in $G_i$. 
	For each $x\in X_i\setminus Y_i^v$, there exists some other boundary vertex $y\in Y^v_i$ such that
	$d_G(v,x)= d_G(v,y)+d_{G_i}(y,x)$. Therefore, for each $u\in V(P)$, we have
	\[
		d_{G}(v,u) ~=~ \min \{ d_G(v,x)+d_{G_i}(x,y)\mid x\in X_i\} 
					      ~=~ \min \{ d_G(v,y)+d_{G_i}(y,u)\mid y\in Y_i^v\} .
	\]
	Because the selection we made for $Y_i^v$ and the uniqueness of shortest paths in $G$,
	we have that
	\begin{equation}
		\forall y,y'\in Y_i^v, y\neq y': ~~~ d_G(v,y) < d_G(v,y')+ d_{G_i}(y',y).
		\label{eq:condition_piece}
	\end{equation}
	Using the shortest-path trees to $x\in X_i$, we can identify the relevant pairs 
	$\{ (v,y)\mid v\in A_i, y \in Y_i^v\}$ in $O(n|X_i| )$ time. 
	Since $\sum_i |X_i| = O(\sqrt{r})$, over all indices $i\in [h]$ we spend $O(nr^{1/2})$ time.
		
	For each $i\in [h]$, fix an embedding of $G_i$ such that 
	$C_i$ defines the outer face and thus $X_i$ lies in the outer face of $G_i$ and $P$.
	
	Now there are slight differences depending on the data we want to compute.
	The difference lies in which data structure we use.
	Let us first consider the problem of computing $\adding(v,U,G)$.
	We apply Theorem~\ref{thm:perpiece_sum} for the graph $G_i$ and the piece $P$ with respect
	to the set $X_i$. 
	Since $G_i$ has $O(n)$ vertices and $P$ has $O(r)$ vertices,	
	the preprocessing takes $\tilde O(n+ |X_i|^3 r^2))$ time.
	Now, for each vertex $v\in A_i$, we 
	consider the graph $G_i^+(Y_i^v)$ with edge weights 
	$\lambda(\dart{x_0}{y})=d_G(v,y)$ for all $y\in Y_i^v$.
	Note that, with these weights, the property in~\eqref{eq:condition_piece}
	corresponds to condition~\eqref{eq:condition}.
	Moreover, for each $u\in V(P)$ we have $d_G(v,u)=d_{G_i^+(Y_i^v)}(x_0,u)$.
	Therefore, we can use the data structure of Theorem~\ref{thm:perpiece_sum} 
	to get in $\tilde O(|Y_i^v|)=\tilde O(|X_i|)$ time
	\begin{align*}
				\adding(x_0,U,G_i^+(Y_i^v))~&=~\sum_{u\in U} d_{G_i^+(Y_i^v)}(x_0,u)
							~=~\sum_{u\in U} d_{G}(v,u)\\
							&=~\adding(v,U,G).
	\end{align*}

	Iterating over all $i\in [h]$ and noting that $A_1,\dots,A_h$ is a partition of $V(G)\setminus V(P)$,
	we obtain the desired values: $\adding(v,U,G)$ for all $v\in V(G)\setminus V(P)$.
	The running time for the preprocessing of this last step is
	\[
		\sum_i \tilde O( n + |X_i|^3 r^2) ~=~ \tilde O(  nh +|X|^3 r^2 )
		~=~ \tilde O(nh + r^{7/2}).
	\]
	and the running time for the queries is
	\[
		\sum_i \tilde O(|A_i| \cdot |X_i|)  ~=~ \tilde O\left( (\sum_i |A_i|) \cdot (\sum_i |X_i|) \right)
		~=~ \tilde O(n r^{1/2}).
	\]
	The result in the first item follows.
	
	For computing $\diam(v,U,G)$, we use the same approach, but employ Theorem~\ref{thm:perpiece_max}
	instead of Theorem~\ref{thm:perpiece_sum}.
	The preprocessing time for $i\in [h]$ has an extra factor $\tilde O(|X_i|^4)$.
	Therefore, the preprocessing time in the last step becomes
	\[
		\sum_i \tilde O( n + |X_i|^3 r^2 + |X_i|^4) ~=~ \tilde O(  nh +|X|^3 r^2 + |X|^4 )
		~=~ \tilde O(nh + r^{7/2}).
	\]
	The rest is essentially the same, and we obtain the claim in the second item.
	
	For computing $\counting(v,U,G,\delta)$, we use the same approach, but employ Theorem~\ref{thm:perpiece_count}
	instead of Theorem~\ref{thm:perpiece_sum}.
	The preprocessing time for $i\in [h]$ has an extra factor $\tilde O(|X_i|^3 r^3)$.
	Therefore, the preprocessing time in the last step becomes
	\[
		\sum_i \tilde O( n + |X_i|^3 r^3 ) ~=~ \tilde O(  nh +|X|^3 r^3)
		~=~ \tilde O(nh + r^{9/2}).
	\]
	The rest is essentially the same, and we obtain the claim in the third item.	
\end{proof}

\paragraph{Working over all pieces.}
We can now obtain our main result. 

\begin{proof}[Proof of Theorem~\ref{thm:main}]
	Adding edges of sufficiently large lengths, 
	we may assume that $G$ is triangulated. 
	We also embed $G$. These operations can be done in linear time.
	With a slight abuse of notation, we keep using $G$ for the resulting
	embedded, triangulated graph.
	
	We compute an $r$-division ${\mathcal P}=\{ P_1,\dots, P_k\}$ 
	of $G$ with few holes, for a parameter $r$ to be specified below. 
	According to Theorem~\ref{thm:division}, this takes $O(n)$ time.
	
	To avoid double counting we assign each vertex to a unique piece, as follows.
	For each vertex $x$ of $G$ we select a unique index $i(x)$ such that
	the piece $P_{i(x)}$ contains $x$. For each piece $P_j\in {\mathcal P}$,
	we define the set $U_j=\{ x\in V(P_j)\mid i(x)=j \}$.
	The sets $U_1,\dots, U_k$ are a partition of $V(G)$ and can be easily computed
	in linear time.
	
	Next, we iterate over the pieces and, for each piece $P_i\in {\mathcal P}$, we use 
	Lemma~\ref{le:outside_piece} to compute the values
	\[
		\adding(v,U_i,G),~\diam(v,U_i,G)
		~~~ \forall v\in V(G)\setminus V(P_i).
	\]
	We also use Lemma~\ref{le:withinpiece} to compute 
	\[
		\adding(v,U_i,G),~\diam(v,U_i,G)
		~~~ \forall v\in V(P_i).
	\]
	Since the piece $P_i$ has $O(1)$ holes, we spend $\tilde O(n + r^{7/2} + nr^{1/2})= \tilde O(r^{7/2} + nr^{1/2})$ time per piece.
	Iterating over the $O(n/r)$ pieces, we get 
	\[
		\adding_G(v,U_i),~\diam_G(v,U_i), ~~~ \forall v\in V(G),~ i\in [k]
	\]
	in time 
	\[
		O(n/r)\cdot \tilde O(r^{7/2} + nr^{1/2}) = \tilde O( nr^{5/2} + n^2 r^{-1/2}).
	\]
	Because $U_1,\dots,U_k$ is a partition of $V(G)$ we can easily compute the desired values because 
	\begin{align*}
		\adding(v,V(G),G) ~&=~ \sum_{i\in [k]}  \adding(v,U_i,G),\\
		\diam(v,V(G),G) ~&=~ \max \{ \diam(v,U_i,G)\mid i\in [k]\}.
	\end{align*}
	(For the diameter of course we do not need that the sets $U_1,\dots,U_k$ are disjoint.)
	Taking $r=n^{1/3}$ the running time becomes $\tilde O(n^{11/6})$ in expectation.
	
	For $\counting(\cdot)$ we use the third item of Lemma~\ref{le:outside_piece}
	to compute for each piece $P_i$
	\[
		\counting(v,U_i,G,\delta)~~~ \forall v\in V(G)\setminus V(P_i).
	\]
	Then, for each piece we spend 
	$\tilde O(n + r^{9/2} + nr^{1/2})= \tilde O(r^{9/2} + nr^{1/2})$.
	Over all pieces, the running time thus becomes
	\[
		O(n/r)\cdot \tilde O(r^{9/2} + nr^{1/2}) = \tilde O( nr^{7/2} + n^2 r^{-1/2}).
	\]	
	Choosing $r=n^{1/4}$ we obtain a running time of $\tilde O(n^{15/8})$.
	Again using that $U_1,\dots,U_k$ is a partition of $V(G)$, we can compute 
	\begin{align*}
		\counting(v,V(G),G,\delta) ~&=~ \sum_{i\in [k]}  \counting(v,U_i,G,\delta).
	\end{align*}
\end{proof}

\begin{corollary}
	Let $G$ be a planar graph with $n$ vertices, real abstract length on its arcs, and no negative cycle.
	In $O(n^{11/6}\polylog(n))$ expected time we can compute the diameter of $G$
	and the sum of the pairwise distances in $G$.
	For a given $\delta\in \RR$,
	in $O(n^{15/8}\polylog(n))$ expected time we can compute 
	the number of pairs of vertices in $G$ at distance at most $\delta$.
\end{corollary}

%%%%%%%%%%%%%%%%%%%%%%%%%%%%%%%%%%%%%%%%%%%%%%%%%%%%%%%%%%%%%%%%%%%%%%%%%%%%%%%%%%%%%%%%%%%%
\section{Discussion}

We have decided to explain the construction through the use
of abstract Voronoi diagrams, instead of providing an 
algorithm tailored to our case. It is not clear to the author
which option would be better. In any case, for people 
familiar with randomized incremental constructions,
it should be clear that the details can be worked out,
once the compact representation of the bisectors using the dual graph is available.
Using a direct algorithm perhaps we could get rid of the 
assumption that the sites have to be in the outer face and
perhaps we could actually build a deterministic algorithm.
In fact, Gawrychowski et al.~\cite{ghmsw18} do follow this path and have
obtained a deterministic algorithm.

There are also deterministic algorithms to compute 
abstract Voronoi diagrams~\cite{Klein89,KleinLN09}.
However, they require additional elementary operations and properties. 
Also, when the abstract Voronoi diagram
has a forest-like shape, it can be computed in linear time~\cite{BohlerKL14}. 
It is unclear to the author whether these results are applicable in our case.

We think that the algorithm can be extended to graphs on surfaces of small genus, but
for this one should take care to extend the construction of abstract Voronoi diagrams to
graphs on surfaces or to 
planar graphs when the sites are in $O(g)$ faces, where $g$ is the genus of the surface.
Let us discuss the reduction to the planar case.
The first step is to find an $r$-division where each part is planar.
For this we can use the separator theorem of Eppstein~\cite{Eppstein03}. 
It computes a set of curves on the surface that pass through $O(\sqrt{gn})$ vertices of $G$ and
do not pass through the interior of any edge. 
Moreover, cutting along them gives a collection of planar patches, possibly with multiple holes.
Taking a maximal subset of the curves that are homologically independent, we will get
$O(g)$ curves that pass through $O(\sqrt{gn})$ vertices and cut the surface into planar patches.
Now we can compute an $r$-division in each of the patches.
We can also compute the distances from all the boundary vertices
because shortest paths in the presence of negative weights can be computed in subquadratic time~\cite{ChambersEN12,LT79}.
Now we run into problems. Consider a vertex $x$ and a piece $P$. The shortest paths from $x$ to different vertices
of $P$ can pass through different boundary cycles. In the planar case, we always have a boundary cycle that intersects
all the paths from $x$ to all vertices of $P$. There is no such cycle in the case of surfaces, for example,
when one of the planar patches is a cylinder. Computing Voronoi diagrams for sites placed in $O(g)$ cycles
would handle this problem. While we think that this should be doable, it requires non-trivial work. In particular,
some of the holes may have non-trivial topology. It seems that the new result by Gawrychowski et al.~\cite{ghmsw18}
may be the missing piece for making this work.

Gawrychowski et al.~\cite{ghmsw18} have managed to reduce our exponent $11/6$ for the diameter to $5/3$.
The author would be surprised if the problems considered in this paper can
be solved in near-linear time. Thus, the author conjectures that there should be 
some conditional lower bounds of the type $\Omega(n^{c})$ for some constant $c>1$.

\section*{Acknowledgments}
This work was initiated at the Dagstuhl seminar Algorithms for Optimization Problems in Planar Graphs, 2016. I am very grateful to Kyle Fox, Shay Mozes, Oren Weimann, and Christian Wulff-Nilsen for several discussions on the problems treated here.
I am also grateful to the reviewers of the paper for their many useful suggestions.

\end{document}